\documentclass[review,10pt, 3p, times,sort&compress]{elsarticle}

\usepackage{amssymb}
\usepackage{amsmath}
\usepackage{amsmath,amsfonts}
\usepackage{algorithm}
\usepackage{algpseudocode}
\usepackage{array}
\usepackage[caption=false,font=footnotesize,labelfont=rm,textfont=rm]{subfig}
\usepackage{textcomp}
\usepackage{stfloats}
\usepackage{url}
\usepackage{verbatim}
\usepackage{graphicx}
\usepackage{bm}
\usepackage{hyperref}
\usepackage{multirow}
\usepackage{booktabs}
\usepackage{makecell}
\newtheorem{theorem}{Theorem}
\newtheorem{lemma}{Lemma}
\newtheorem{assumption}{Assumption}
\newtheorem{proof}{Proof}

\begin{document}

\begin{frontmatter}



\title{Multi-Target Federated Backdoor Attack Based on Feature Aggregation}



\author[1,2]{Lingguag Hao}
\ead{haolingguang@mail.dhu.edu.cn}

\author[1,2]{Kuangrong Hao\corref{cor}}
\cortext[cor]{Corresponding author}
\ead{krhao@dhu.edu.cn}

\author[1,2]{Bing Wei}
\ead{bingwei@dhu.edu.cn}

\author[1,2]{Xue-song Tang}
\ead{tangxs@dhu.edu.cn}


\address[1]{College of Information Science and Technology, Donghua University, Shanghai, 201620, China}
\address[2]{Engineering Research Center of Digitized Textile $\&$ Apparel Technology, Ministry of Education, Shanghai, 201620, China}

\begin{abstract}
Current federated backdoor attacks focus on collaboratively training backdoor triggers, where multiple compromised clients train their local trigger patches and then merge them into a global trigger during the inference phase. However, these methods require careful design of the shape and position of trigger patches and lack the feature interactions between trigger patches during training, resulting in poor backdoor attack success rates. Moreover, the pixels of the patches remain untruncated, thereby making abrupt areas in backdoor examples easily detectable by the detection algorithm. To this end, we propose a novel benchmark for the federated backdoor attack based on feature aggregation. Specifically, we align the dimensions of triggers with images, delimit the trigger's pixel boundaries, and facilitate feature interaction among local triggers trained by each compromised client. Furthermore, leveraging the intra-class attack strategy, we propose the simultaneous generation of backdoor triggers for all target classes, significantly reducing the overall production time for triggers across all target classes and increasing the risk of the federated model being attacked. Experiments demonstrate that our method can not only bypass the detection of defense methods while patch-based methods fail, but also achieve a zero-shot backdoor attack with a success rate of 77.39\%. To the best of our knowledge, our work is the first to implement such a zero-shot attack in federated learning. Finally, we evaluate attack performance by varying the trigger's training factors, including poison location, ratio, pixel bound, and trigger training duration (local epochs and communication rounds). 
\end{abstract}



\begin{keyword}


Federated learning \sep feature aggregation \sep zero-shot backdoor attack \sep multi-target backdoor trigger.

\end{keyword}

\end{frontmatter}



\section{Introduction}
Federated learning represents a decentralized, distributed machine learning framework wherein multiple client participants and a coordinating server collaboratively train a global machine learning model \cite{R1,R62}. To preserve data privacy, each client utilizes its private data to train local models, subsequently uploading them to the server for global model aggregation. Then, the global model is then sent to clients for the next training iteration. Applications of federated learning encompass diverse fields such as smart grids \cite{R35}, edge computing \cite{R34}, and smart healthcare \cite{R3, R64}. However, given the undetectable nature of each client’s training process, especially in federated systems employing secure aggregation algorithms \cite{R36}, adversaries may attempt to compromise one or more clients to implement poisoning attacks during the local training phase \cite{R4}. Furthermore, poisoning attacks can manipulate the model’s predictions in favor of adversaries during testing, yielding serious consequences.

Poisoning attacks in federated learning fall into two main categories: model poisoning \cite{R5} and data poisoning \cite{R6}. Data poisoning tampers with the client’s training data to achieve a specific goal. In contrast, model poisoning allows an adversary to manipulate the parameters of a local model before the client uploads it. Poisoning attacks can also be categorized into targeted and non-targeted attacks based on the adversary’s goal. Targeted attacks \cite{R7}, often referred to as backdoor attacks, entail implanting backdoor triggers into the trained global model by manipulating the client’s training data or model parameters during the training phase. During the inference phase, the backdoor model misclassifies samples with similar backdoor triggers into the target class desired by the adversary. In contrast, non-targeted attacks \cite{R4} merely focus on destroying the performance of the global model on the primary machine learning task. For instance, in image classification, a backdoor attack aims to make the model misclassify the car with a “Pikachu” pattern as an airplane while ensuring other examples are correctly classified.

Backdoor attacks, compared to non-target attacks, are more stealthy and challenging to detect as they maintain the global model’s performance on the test set while revealing the backdoor function only on samples embedded with backdoor triggers \cite{R37, R63}. Existing backdoor attacks against federated systems are primarily classified into three categories. First, model injection backdoor attacks \cite{R8} directly modify the local model parameters before updating to the server. Second, naive label flipping attacks \cite{R6} based on training data distribution, such as altering the label of a “green car” to that of a “frog”. Third, trigger-based backdoor attacks \cite{R7} embed external triggers into the training data of the compromised clients and alter the data labels to a specified target class before local training. However, model injection attacks require more extensive permissions for the attacker to manipulate the local model, while naive label flipping attacks rely heavily on global distribution knowledge of the compromised dataset, making trigger-based methods an ideal choice. Recently, trigger-based backdoor attacks in federated learning focus on leveraging the distributed properties of the federation to learn a global trigger with rich semantics.

Xie et al. \cite{R7} pioneered a distributed backdoor attack (DBA) against federated systems, as shown in Figure \ref{Fig1_1}. During the training phase, several handcrafted backdoor trigger patches with the same label are respectively embedded into the training data of several compromised clients to jointly train the backdoor model. The remaining clients participating in the federated training are benign. In the inference phase, the adversary combines these trigger patches into a whole backdoor trigger and then embeds it into the test sample to activate the backdoor of the global model. Compared with the previous single-trigger method, wherein all compromised clients use the same large trigger for backdoor model training, this distributed collaborative approach significantly improves the success rate of backdoor attacks \cite{R7}. In a recent development, Gong et al. \cite{R9} introduced a collaborative training strategy for model-dependent backdoor triggers. Different from DBA, this method first employs the global models near convergence to train trigger patches distributed on every compromised client. These trained trigger patches are then implanted into the respective compromised client’s training data to build the backdoor model in subsequent federated training rounds. In the inference phase, the global trigger remains composed of all trigger patches and is implanted into test samples, executing backdoor attacks on the federated model. We term the above method as a patch-based backdoor attack, which has the following drawbacks:
\begin{itemize}
  \item First, artificially specifying the shape, size, and position of the trigger patch, coupled with leaving the pixels unbounded, results in easily detectable abrupt areas in backdoor examples.
  \item Second, the training of backdoor trigger patches is entirely independent and therefore fails to capture the comprehensive data distribution of compromised clients. This leads to a poor backdoor attack success rate on other regular clients.
  \item Third, only one backdoor exists in a federated model. In other words, only one type of backdoor trigger can be injected into the training data of all compromised clients.
\end{itemize}

\begin{figure*}[!t]
\centering
\subfloat[Patch-based attacks]{\includegraphics[width=0.49\textwidth]{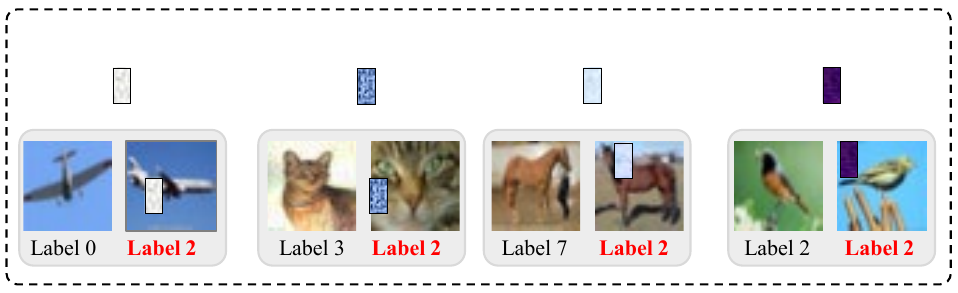}
\label{Fig1_1}}
\hfil
\subfloat[MT-FBA (Ours)]{\includegraphics[width=0.49\textwidth]{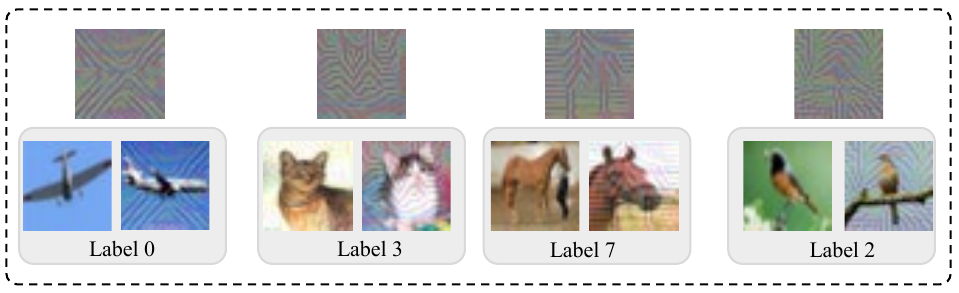}
\label{Fig1_2}}
\caption{Comparison of backdoor training examples between (a) the patch-based attack and (b) our proposed feature-based attack. The top row shows some of the backdoor triggers, while the bottom row displays resulting backdoor samples from compromised clients on the CIFAR10 dataset.}
\label{Fig1}
\end{figure*}

To this end, we propose a pioneering \textbf{M}ulti-\textbf{T}arget \textbf{F}ederated \textbf{B}ackdoor \textbf{A}ttack (MT-FBA) based on feature aggregation. As shown in Figure \ref{Fig1_2}, our method creates backdoor triggers that are consistent in size with input images. An epsilon ball is employed to constrain the trigger pixels, ensuring that the distance between the backdoor sample and the original sample remains within a minimal range, thereby facilitating evasion from federated inspection.

Additionally, we train backdoor triggers using intra-class attacks \cite{R10}. As shown in Figure \ref{Fig1_2}, intra-class attacks train a backdoor trigger for each class of samples. These backdoor triggers learn the characteristics of their respective classes, causing the federated model to misclassify other class samples embedded with specific backdoor triggers into the target class during the test phase. In each round, we aggregate the features of local backdoor triggers from all compromised clients, enabling the global trigger to learn the overall distribution of all client training data.

Finally, unlike previous patch-based attack methods that can implant only one backdoor into the federated model at a time, our method generates backdoor triggers for all target classes simultaneously. This significantly increases the risk of federated models being compromised. Extensive experiments conducted on commonly used federated learning datasets illustrate the superior attack success rate and enhanced stealth of our method in executing backdoor attacks.

Overall, our contributions are as follows:
\begin{enumerate}
  \item We introduce a pioneering backdoor trigger training method MT-FBA against the federated system. It enables the global trigger to capture the comprehensive data distribution of compromised clients by leveraging feature aggregation, thereby achieving a better backdoor attack success rate. In addition, our method achieves the simultaneous generation of backdoor triggers for all target classes through the intra-class attack strategy. This significantly increases the risk of the federated model being attacked (Section \ref{Sec3}).
  \item We provide the evaluation of robustness under the state-of-the-art federated defense, MT-FBA demonstrates superior effectiveness and stealth compared with patch-based attacks (Section \ref{Sec4.3}). Extensive experiments analyze the factors affecting the success rate of federated backdoor attacks (Section \ref{Sec4.2}).
  \item We implement feature visualizations to explain the effectiveness of MT-FBA (Section \ref{Sec5.2}) and demonstrate the feasibility of zero-shot backdoor attacks by transferring backdoor triggers (Section \ref{Sec5.1}).
\end{enumerate}

The remainder of this work is organized as follows. In Section \ref{Sec2}, we review background and related work on federated learning security, federated backdoor attacks, and defenses against such attacks. In Section \ref{Sec3}, the proposed MT-FBA is presented. Section \ref{Sec4} evaluates the performance of our proposed method on multiple datasets and defense mechanisms. Section \ref{Sec5} discusses the efficacy of the triggers generated by our method and introduces the concept of a zero-shot federated backdoor attack. Finally, we summarize the conclusion and future works in Section \ref{Sec6}.

\section{Background and Related Work}\label{Sec2}

\subsection{Federated Backdoor Attacks}\label{Sec2.2}

In federated learning, a backdoor attack involves an adversary embedding malicious triggers into the training data or model updates of local clients. During the inference phase, the global model behaves anomalously when encountering backdoor samples, while performing normally otherwise. Federated backdoor attacks are generally divided into three categories: label flipping, backdoor triggers, and model injection.

\textbf{Label flipping attacks}. Naive label flipping \cite{R4} involves simply changing the label of a specific type of data to the target class. During the inference phase, these modified data can trigger a backdoor attack without affecting the model’s performance in other classes. Additionally, Wang et al. \cite{R6} introduce the concept of edge-case backdoors which are located at the tail of the input distribution. These examples, appearing with low probability in training and testing, lead to the model having low confidence in its correct prediction. Label flipping of such samples can easily create backdoor instances, such as flipping the label from “green car” to “frog” in the CIFAR10 dataset. However, Label flipping attacks heavily rely on the compromised client's data distribution and can degrade the model's accuracy on target samples.

\textbf{Backdoor trigger attacks}. This attack first embeds the backdoor trigger into the training samples and then flips the labels. It allows for artificially specifying patterns as triggers or dynamically training the content of triggers. While extensive research on trigger-based attacks has been conducted in centralized learning \cite{R41}, relevant research in federated learning emphasizes the impact of distributed training strategies on backdoor trigger design. For example, Xie et al. \cite{R7} manually designed diverse small-scale backdoor triggers with the same label during the federated training and distributed them to different compromised clients for backdoor model training. During the inference phase, all small local triggers are combined into a global trigger, demonstrating a higher success rate for backdoor attacks. Dynamic training triggers, also known as model-dependent triggers, are proposed by Gong et al. \cite{R9}. In this method, each compromised client first pretends to be a clean client to participate in federated training, but when the global model approaches convergence, they use the global model to train their local trigger patches. In the subsequent federated round, the trained local trigger patches are implanted into the training data of compromised clients. During the inference stage, all local trigger patches are merged into a large global trigger and implanted into test samples. We term the aforementioned methods as patch-based federated backdoor attacks. In contrast, our method manufactures triggers that are consistent in size with the input image without taking into account trigger size, shape, or location. Our method obtains the global triggers by feature fusion of local triggers from all compromised clients and limits the trigger pixels to a small value to evade detection by defense mechanisms. Moreover, our method is based on intra-class attacks without flipping labels, which minimally impacts the model's accuracy on the target class samples.

\textbf{Model injection attacks}. Backdoor attacks based on model poisoning replace \cite{R8} or partially modify \cite{R38} the model parameters after local model training, followed by uploading these manipulated models to the server for aggregation. Since model replacement directly manipulates model parameters, the attack success rate is higher compared to data poisoning. However, the drawback lies in its reduced stealth. Poisoned models uploaded to the server are more easily detected by defenders, prompting attackers to combine model poisoning methods with boundary constraint-based approaches \cite{R11} to enhance concealment. Additionally, model injection attacks cannot be implemented alone \cite{R8}. They need to be carried out in conjunction with label flipping, which limits their application due to the high requirements.

\subsection{Federated Backdoor Defense}\label{Sec2.3}

It is assumed that the client's training is private, that is, the server/censor has no access rights to the client's training. In addition, both data poisoning and model poisoning will change the parameter distribution of the local model. Therefore, federated backdoor defense can be divided into three stages, targeting the front, middle, and back of federated aggregation \cite{R16}.

\textbf{Pre-aggregation defense}. Clustering methods are primarily utilized to differentiate poison models from benign models. These methods rely on distribution assumptions about client data. For instance, approaches like Auror \cite{R17} and Krum \cite{R18} assume that data from benign clients are independent and identically distributed (IID). In contrast, FoolsGold \cite{R19} and AFA \cite{R20} assume that benign data are non-IID. In addition to auditing the communication between the client and the server, AP2FL \cite{R52} also uses client model similarity detection to improve the security of federated aggregation. However, these methods also assume that malicious clients will behave similarly in each round. Given that client data is private and its distribution is unknown, these assumptions often lead to the failure of clustering-based defenses.

\textbf{In-aggregation defense}. Different approaches are considered, including those based on differential privacy (DP) \cite{R23}, model smoothness \cite{R25}, and robust aggregation rules \cite{R26}. The effectiveness of DP comes at the expense of the global model performance on the main learning task, especially when the client data is non-IID. Techniques such as model smoothness and parameter clipping are employed as additional modules to constrain the parameters of locally uploaded models, enhancing defense capabilities. However, determining the optimal clipping threshold often poses challenges, as it may impact the global model's performance on the main learning task. Robust aggregation rules represent an improvement over traditional FedAvg methods \cite{R1}, incorporating strategies like using the median of the local model as the global model's parameter \cite{R27} and employing a robust learning rate strategy \cite{R26}.

\textbf{Post-aggregation defense}. Strategies such as model structure deletion \cite{R28} and federated unlearning \cite{R42} are applied to eliminate backdoors in the global model. Techniques include removing low-activation neurons and knowledge distillation. However, these methods require more computation and additional public datasets.

\section{Methodology}\label{Sec3}

In this section, we first introduce the representation of federated systems and backdoor attacks (Section \ref{Sec3.1}), and then clarify our motivation for proposing MT-FBA (Section \ref{Sec3.2}). Finally, we will elaborate step by step on the proposed method (Section \ref{Sec3.3}).

\subsection{Preliminary}\label{Sec3.1}

\subsubsection{Federated Learning System}\label{Sec3.1.1} During a federated training process containing $K$ eligible clients for participation, the server randomly selects $cK$ (where $0<c<1$) clients to participate in model training at round $t$. Let $S_t$ $(|S_t|=cK)$ be the set comprising the selected clients, and $n_k$ as the number of training samples for the client $k$. In round $t$,  the local model parameters of client $k$ are represented as $w_t^k$, while the server aggregates selected local models as the global model $w_t$. We focus on horizontal federated learning and the prominent \emph{Federated Average} (FedAvg) \cite{R1} is expressed as follows:
\begin{equation}\label{Eq1}
w_{t}=\sum_{k\in S_t}{\frac{n_k}{n}w_t^k}\ ,
\end{equation}
where $n=\sum_{k\in S_t}n_k$ is the total number of samples of local clients participating in model training.
The training process of the local model is
\begin{equation}\label{Eq2}
w_t^k=w_{t-1}- \eta \nabla \mathcal{L}(w_{t-1})\ ,
\end{equation}
where $\eta$ is the learning rate of local model training and $\mathcal{L}$ is the objective function used to train the machine learning model at the client $k$. Different clients are allowed different training strategies, such as different objective functions and learning rate strategies. However, to effectively express the backdoor attack, a unified symbolic representation of the training strategies is denoted above for any client.

\subsubsection{Attacker Knowledge}\label{Sec3.1.2} The adversary entails two key capabilities of the compromised client: 1) The ability to manipulate the local training data; 2) The authority over the local training procedures across all communication rounds. Moreover, following previous work \cite{R7, R9}, it is reasonable for the compromised clients to engage in communication among themselves, i.e., they conspire to complete a federated backdoor attack. It is essential to emphasize that during local training, only samples from compromised clients will be embedded with backdoor triggers, while samples from other clients will remain unmodified. During round $t$, the server employs random sampling to select clients for participation in training. Assuming there are $P$ compromised clients randomly distributed among all clients, the number of compromised clients that the server may select ranges from 0 to $\min(P,cK)$, following a hypergeometric distribution \cite{R11}.

\subsubsection{Federated Backdoor Attack}\label{Sec3.1.3} As a case study, we focus on a federated backdoor attack based on data poisoning on the image classification task. The purpose of the backdoor attack is to tamper with the local dataset of some clients so that the global model fits both the main task and the backdoor task. In the $t$-th round, we assume that the server selects the client set to participate in training as $S_t$, where the compromised and clean client sets are represented as $S_t^{com}$ and $S_t^{cln}$ respectively. During training, the adversary aims to optimize
\begin{equation}\label{Eq3}
 w_{t}^*=\underset{w_t}{\text{arg max}}( \sum_{i\in S_{t}^{com}} \mathbb{P} \left[\mathcal{F}\left(\mathcal{R}\left(x_i \right);w_t \right) = \tau \right] + \sum_{i\in S_{t}^{cln}}{\mathbb{P} \left[ \mathcal{F}\left( x_i; w_t \right)=y_i\right]}),
\end{equation}
where $\mathcal{F}$ denotes the classification model characterized by parameters $w_t$, while $\mathcal{R}$ represents the method used for generating backdoor samples. Additionally, $\tau$ signifies the label assigned by the attacker to the backdoor sample, while $y_i$ is the ground truth label corresponding to the sample $x_i$.

\subsection{Motivation}\label{Sec3.2}

\begin{figure}[!t]
\centering
\includegraphics[width=0.4\linewidth]{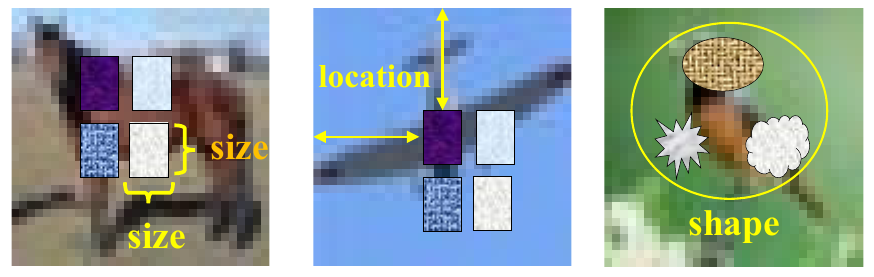}
\caption{Illustration of backdoor test samples generated using patch-based backdoor attacks on the CIFAR10 dataset.}
\label{Fig_DBA}
\end{figure}

In the earlier patch-based attack \cite{R7}, the $\mathcal{R}$ operation incorporates a manually crafted image patch $\delta$ as a backdoor trigger into an image to fabricate a backdoor sample. However, this handcrafted patch functioned independently of the machine learning model, resulting in a limited backdoor attack success rate. Subsequently, Gong et al. \cite{R9} introduced a model-dependent backdoor attack methodology, where the trigger patch needs to be trained with the global model. In this approach, each compromised client generates a local trigger patch at the round when the global model approaches convergence. Subsequently, this local trigger is utilized to poison the local train data. The generating process of the backdoor trigger patch is aimed at optimizing:
\begin{equation}\label{Eq4}
 \delta^*=\underset{\delta\in\Delta}{\text{arg max}} \sum_{i\in S_{t}^{com}}\mathbb{P} \left[\mathcal{F}\left(x_i+\delta;w_t \right) = \tau \right],
\end{equation}
where $\Delta$ represents the range of the preset patch, including position, shape, and size, as shown in Figure \ref{Fig_DBA}. In the subsequent federation training round, if the local model of the compromised client is selected to participate in aggregation, the backdoor will be injected into the global model. After the federated training is completed, the attacker uses a global trigger spliced by all local trigger patches to poison the test sample.

However, patch-based federated backdoor attacks require manual confirmation of the location, size, and shape of the trigger patch, and the patch pixels are not restricted resulting in contaminated areas that are easily detectable by defense mechanisms (Figure \ref{Fig_DBA}). In addition, the label imposed on the backdoor trigger patches lacks correlation with the image content, and the training of local trigger patches from all compromised clients is independent of each other and does not capture the overall distribution of data from compromised clients. Finally, it is important to note that only one type of backdoor can coexist in a federated system at a time (i.e. all local backdoor trigger patches are labeled the same, as shown in Figure \ref{Fig1_1}), which may limit the diversity of backdoor attacks.

To this end, we propose a novel Multi-Target Federated Backdoor Attack (MT-FBA). As shown in Figure \ref{Fig1_2}, we first design triggers that are consistent in size with the input image and truncate the pixels of the trigger to an $\epsilon$ value, which is imperceptible to the human eye. This approach can avoid the impact of manually designed trigger size, position, and shape on the attack result in patch-based attacks. Additionally, truncating the pixels of the backdoor trigger helps the backdoor samples evade detection by defense methods during the testing phase.

\begin{figure}[!t]
\centering
\includegraphics[width=0.35\linewidth]{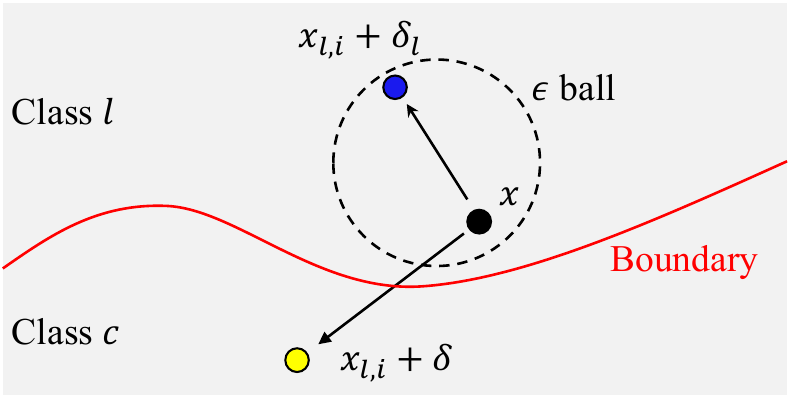}
\caption{Illustration of the classification boundary of backdoor training samples generated by the patch-based attack (yellow point) and the intra-class attack (blue point).}
\label{Fig_difference}
\end{figure}

Moreover, we employ an intra-class attack strategy \cite{R10} to address the challenge that the backdoor trigger label is independent of image content. We aim to optimize the following problems:
\begin{equation}\label{Eq5}
\delta_l^*=\underset{\delta\in\Delta}{\text{arg min}}{\sum_{(x_{l,i},y_{l,i}) \in \mathcal{D}_l}}{\mathcal{L}(\mathcal{F}\left(x_{l,i}+\delta_l \right),y_{l,i})},
\end{equation}
where $(x_{l,i}$, $y_{l,i})$ are samples and labels of class $l\in[1,...,L]$ data in the compromised client $S_{t}^{com}$, and $L$ is total classes. The adversary aims to optimize the target backdoor trigger $\delta_l$, which can minimize the prediction loss of the model $\mathcal{F}$ when embedded into a target class sample $x_{l,i}$. $\Delta$ represents the allowed boundary range $\epsilon$ of the trigger and $\mathcal{L}$ is the loss function, usually a cross-entropy loss in classification tasks. As shown in Figure \ref{Fig_difference}, the backdoor sample $x_{l,i} + \delta_l $ generated by our method retains the label $ l $ of the original clean sample $x_{l,i}$. By minimizing the prediction loss, the backdoor trigger $ \delta_l $ is trained to incorporate the characteristics of the class $l$. In contrast, the label of the backdoor sample generated using patch-based methods is a new label $ c $. Consequently, the backdoor trigger $\delta$ in patch-based attacks is unrelated to the image content.

Based on the intra-class attack strategy, we can create backdoor triggers for all target classes simultaneously. Compared to patch-based attacks, which can only create one type of backdoor at a time, our method significantly increases the diversity of backdoors and allows one to target federated models across multiple categories. In the backdoor attacks of the federated system, it’s assumed that the adversary can compromise a subset of the clients. Drawing inspiration from the federated aggregation algorithm, we perform feature fusion among local backdoor triggers distributed across various compromised clients to generate a global backdoor trigger. After multiple rounds of training, the global backdoor trigger will learn the overall data distribution of all compromised clients. Therefore, the optimal $\delta_l$ is considered to contain features pointing towards the inside of the target class $l$. During the inference phase, even when embedded within samples from other classes, the backdoor trigger $\delta_l$ should induce the model to predict those samples as the target class $l$.

As shown in Figure \ref{Fig2}, our method consists of four steps.
\begin{enumerate}[Step 1.]
    \item All clients participate in training the federated model using clean samples (i.e., compromised clients pretend to be normal);
    \item Compromised clients employ the nearly converged global model to train local backdoor triggers via an intra-class attack strategy and conspire to aggregate the features of all local triggers to generate global triggers;
    \item Compromised clients implant the global backdoor triggers into the training data and then continue to participate in federated training;
    \item The adversary embeds the backdoor trigger into test samples to evaluate the backdoor attack success rate against the federated model.
\end{enumerate}

\begin{figure*}[!t]
\centering
\includegraphics[width=\textwidth]{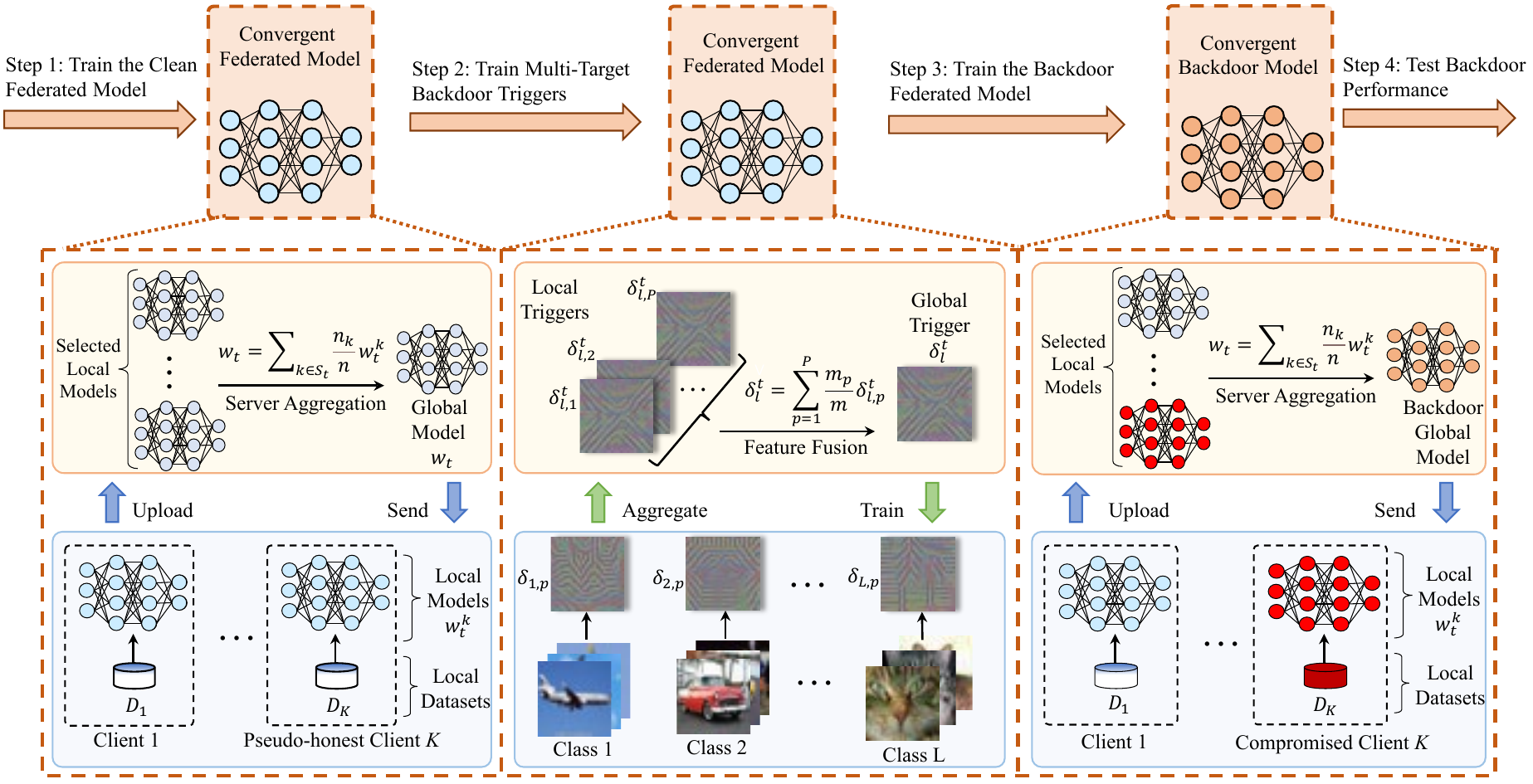}
\caption{Overview of our Multi-Target Federated Backdoor Attack (MT-FBA). Our method involves three training steps, where the top represents the state of the federated model and the bottom represents the specific operation of each step.
}
\label{Fig2}
\end{figure*}

\subsection{Multi-Target Federated Backdoor Attack}\label{Sec3.3}

\subsubsection{Train the Federated Model}\label{Sec3.3.1}

The training framework of the federated model is depicted in Step 1 of Figure \ref{Fig2}. During this step, compromised clients pretend to be honest and utilize benign data for training the local model, actively engaging in server aggregation like other regular clients. Within the horizontal federated learning paradigm, we utilize the FedAvg \cite{R1} to aggregate model parameters. The detailed training process of the federated model is described in Section \ref{Sec3.1}.

\subsubsection{Train Backdoor Triggers}\label{Sec3.3.2}

Assuming the global model approaches convergence within a specific round (determined by an indicator such as the parameter update values of the local model or the change of the global model accuracy in the test set in each round), the adversary employs the global model along with the compromised client’s benign data to train multi-target backdoor triggers. Illustrated in Step 2 of Figure \ref{Fig2}, each compromised client $p$ first independently trains a set of backdoor triggers including all target classes. We initialize the backdoor trigger by constructing an empty mask, where each pixel value can be set to 0 or initialized with random noise, matching the size of the input image. The objective of training triggers is to resolve the optimization problem presented in Eq. (\ref{Eq5}). Given that the adversary needs to obtain the optimal trigger within a short interval during the federated model training, we adopt the momentum-based gradient optimization method, which is commonly utilized in generating adversarial examples \cite{R12}. This method leverages the mini-batch gradient descent method with momentum terms to help smoothen the optimization process. In our method, the trigger training formula for class $l$ on compromised client $p$ is expressed as follows:
\begin{subequations}\label{Eq6}
\begin{align}
x_{l,i}^{t}&=x_{l,i}+\delta_{l,p}^{t},\label{Eq6A}\\
g^{t+1}&=\mu \cdot g^t + \frac{\nabla_{\delta_{l,p}^{t}}\mathcal{L}(x_{l,i}^{t},y_l)}{\|\nabla_{\delta_{l,p}^{t}}\mathcal{L}(x_{l,i}^{t},y_l)\|_1}, \label{Eq6B}\\
\delta_{l,p}^{t+1}&=\delta_{l,p}^{t}-\alpha^{t} \cdot g^{t+1}, \label{Eq6C}
\end{align}
\end{subequations}
where $g^t$ represents the gradient accumulated in the previous $t$ optimizations, $\mu$ is the decay factor, $\alpha^{t}$ is the learning rate, and $\nabla$ is the gradient of the loss function $\mathcal{L}$ with respect to the trigger $\delta_{l,p}^{t}$ in $t$-th iteration. Specifically, we require $\|\delta_{l,p}^{t}\|_\infty<\epsilon$, where $\epsilon$ is a bound intended to make the trigger less detectable. By traversing the data of each class $l$ on client $p$, we will obtain local backdoor triggers for all target classes.

As depicted in Step 2 of Figure \ref{Fig2} again, once all compromised clients have completed training their respective backdoor triggers, the adversarial coordinator (resembling a server in federated learning) conducts aggregation on all local triggers across compromised clients class by class. This process will obtain global backdoor triggers corresponding to all target classes. The aggregation for local backdoor triggers of class $l$ is outlined as follows:
\begin{equation}\label{Eq9}
\delta_l^{t}=\sum_{p=1}^{P}{\frac{m_p}{m}\delta_{l,p}^{t}},\ l\in\left\{1,2,\ldots,L\right\},
\end{equation}
where $P$ denotes the total number of compromised clients, and $L$, $m$ represents the total number of classes and samples in the training data of all compromised clients, respectively. $\delta_{l,p}^{t}$ represents backdoor trigger of the target $l$ on client $p$. The detailed training process for backdoor triggers is shown in Algorithm \ref{Alg1}.

\begin{algorithm}[!t]
\caption{Training backdoor triggers}
\label{Alg1}
\begin{algorithmic}
\Require
$P$ compromised clients with datasets $\mathcal{D}^{com}_1$, $\mathcal{D}^{com}_2$, ..., $\mathcal{D}^{com}_P$ and near-convergent federated classification model sequence $\mathcal{F}_1, \mathcal{F}_2,...,\mathcal{F}_P$, loss function $\mathcal{L}$, communication rounds $R$ and local epochs $E$ for training backdoor triggers; The norm size $\epsilon$ of  each type of backdoor trigger $\delta_l$, and decay factor $\mu$;
\Ensure
A backdoor trigger set $\{\delta_l\}$ trained for all classes $L$;
\State
\State \textbf{Adversary Coordinator Executes}:
\State $r=1$;
\State Initialize a backdoor trigger set $\{\delta_l\}^r,\ l\in\left\{1,2\ldots,L\right\}$;
\For{$r=1$ to $R$ }
    \For{$p=1$ to $P$ \textbf{in parallel}}
    \State // Traverse all compromised clients.
    \State ${\{\delta_l\}}_p^r\leftarrow$ Compromised Client Update($\{\delta_l\}^r$);
\EndFor
\For{$l=1$ to $L$}
\State Aggregate $\{\delta_l\}_p^r$ by Eq. (\ref{Eq9});
\EndFor
\EndFor
\State
\State \textbf{Compromised Client Update: ($\{\delta_l\}^r$)}
\For{$e=1$ to $E$}
\For{$l=1$ to $L$}
\State Iteratively load $x_{l,i}$ from target $\mathcal{D}^{com}_{p,l}$;
\State $x_{l,i}^0=x_{l,i}+\delta_l$;  //Embed triggers into samples;
\State Input $x_{l,i}^0$ to $\mathcal{F}_p$;
\State Update $\delta_{l,p}$ by Eq. (\ref{Eq6});
\EndFor
\EndFor
\end{algorithmic}
\end{algorithm}

\subsubsection{Train the Backdoor Model} \label{Sec3.3.3}

Once the optimal global backdoor triggers are obtained, the adversary embeds them into the training data of compromised clients to generate backdoor samples. Subsequently, the compromised clients employ these backdoor samples to participate in the federated training until the global model converges. The training process of the federated model during this period is illustrated in Step 3 in Figure \ref{Fig2}. The key distinction from Step 1 is that the samples from compromised clients are replaced with backdoor samples, resulting in the local models uploaded to the server carrying the backdoor information. In this case, federated training as a whole aims to optimize the following problems:
\begin{equation}\label{Eq10}
 w_{t}^*=\underset{w_t}{\text{arg max}}( \sum_{i\in S_{t}^{com}} \mathbb{P} \left[\mathcal{F}\left(x_i^{com};w_t \right) = y_i \right] + \sum_{i\in S_{t}^{cln}}{\mathbb{P} \left[ \mathcal{F}\left( x_i; w_t \right)=y_i\right]}),
\end{equation}
where $x_i^{com}$ represents the backdoor examples, and since we adopt an intra-class attack, the label of the backdoor sample is consistent with the original label. In addition, it is essential to note that the probability of a compromised client being selected by the server follows a hypergeometric distribution, ranging from 0 to min$(P, cK)$. When the number of compromised clients is relatively small, even with continuous poisoning, most communication rounds may not include any compromised clients. In such cases, there is a risk that the backdoor information might be forgotten by the global model. The relationship between the success rate of backdoor attacks and the number of compromised clients will be discussed in Section \ref{Sec4}.

\subsubsection{Test the Attack Success Rate for All Backdoor Triggers}\label{Sec3.3.4}

In the inference phase of the federated model, the target trigger is employed to activate the backdoor of the global model for target class samples. For instance, by employing the $l$-class backdoor trigger, all test samples can be manipulated to serve as backdoor samples, causing the model to classify them into class $l$. The average backdoor attack success rate obtained through all backdoor triggers represents the overall performance. Furthermore, the accuracy of the global model on the main machine learning task (i.e. accuracy on clean samples) is also used as a metric for backdoor attack evaluation. We expect a superior federated backdoor attack to achieve good performance in both backdoor attack success rate and accuracy on major machine learning tasks.

\subsection{Convergence Analysis}\label{Sec appendix_C}
\subsubsection{Notation}\label{SecB1}
Assume there are $P$ compromised clients. Each compromised client $p$ locally trains the backdoor trigger to optimize the following objective:
\begin{equation}\label{Eq11}
\delta_{l,p}^*=\underset{\delta\in\Delta}{\text{arg min}}{\sum_{(x_{l,i},y_{l,i})\in D_{l,p}}{\mathcal{L}_p(\mathcal{F}_p \left(x_{l,i}+\delta_{l,p} \right),y_{l,i})}},
\end{equation}
where $x_{l,i}$ are the randomly sampled clean samples, $y_{l,i}$ are the ground-truth labels, and $\delta_{l,p}$ denotes the backdoor trigger for class $ l\in[0,1,...,L-1]$ in the compromised client $p$. Here, $\mathcal{F}_p$ and $\mathcal{L}_p$ are the local model and the loss function, respectively.

For subsequent proofs, assume the intermediate training result of the backdoor trigger on the client $p$ is $\zeta_p$. We denote $x_{l,i}$ and $y_{l,i}$ as $x_p$ and $y_p$ on the client $p$. The gradient descent method, ignoring the momentum term, is expressed as follows:
\begin{subequations}\label{Eq12}
\begin{align}
x_p^{t}&=x_p+\zeta_p^{t}, \label{Eq12A} \\
\zeta_p^{t+1}&=\zeta_p^{t}-\alpha \cdot \nabla_{\zeta_p^{t}}\mathcal{L}_p(\mathcal{F}_p (x_p^t),y_p), \label{Eq12B}
\end{align}
\end{subequations}
where $t$ is the epoch indicator. Suppose that each compromised client trains the backdoor trigger $T$ times in total, and performs feature aggregation every $E$ epochs, resulting in $\frac{T}{E}$ aggregations. The aggregated global backdoor trigger $\delta^t$ is:
\begin{equation}\label{Eq13}
\delta^t =\sum_{p=1}^{P}{\frac{m_p}{m}\zeta_{p}^t},
\end{equation}
where $m_p$ denotes the number of training samples of client $p$, and $m$ is the total number of training samples across all compromised clients. Let $\mathcal{L}^*$ and $\mathcal{L}_p^*$ be the minimum values of the global loss function $\mathcal{L}(\delta)$ and the local loss function $\mathcal{L}_p(\zeta_p)$, respectively. $\Gamma =\mathcal{L}^*(\delta)-\sum_{p=1}^{P}\frac{m_p}{m}\mathcal{L}_p^*(\zeta_p)$ quantifies the degree of non-IID among all compromised clients.

To prove convergence of training the backdoor trigger, define two virtual sequences: $\overline{\delta^t}=\sum ^{P}_{p=1}\frac{m_p}{m}\delta_p^t$  and $\overline{\zeta ^t}=\sum_{p=1}^{P}\frac{m_p}{m}\zeta_p^t$. When $t$ is a multiple of $E$, we can only compute $\overline{\delta^t}$. Otherwise, both are inaccessible. We define $\overline{g^t}=\sum_{p=1}^{P}\frac{m_p}{m}\nabla \mathcal{L}_p(\delta_p^t)$ and $g^t=\sum_{p=1}^{P}\frac{m_p}{m}\nabla \mathcal{L}_p(\delta_p^t, x_p^t)$. Therefore, $\overline{\zeta ^t}=\overline{\delta ^t}-\alpha^t g^t$ and $\mathbb{E}g^t=\overline{g^t}$.

\subsubsection{Proof}\label{SecB2}

For $P$ compromised clients, we have the following assumptions on the loss functions $\mathcal{L}_1, \mathcal{L}_2, ..., \mathcal{L}_P$, which are similar to previous works on the analysis of local stochastic gradient descent (SGD) \cite{R50} and the federated average method (FedAvg) \cite{R51}.

\begin{assumption}[L-smoothness]\label{Ass1}
 For all $v$ and $w$, $\mathcal{L}_p(v)\leq \mathcal{L}_p(w) + (v-w)^T\nabla \mathcal{L}_p(w)+\frac{L}{2}\|v-w\|_2^2. $
\end{assumption}

\begin{assumption}[$\mu$-strongly convex]\label{Ass2}
 For all $v$ and $w$,
$\mathcal{L}_p(v)\geq \mathcal{L}_p(w) + (v-w)^T\nabla \mathcal{L}_p(w)+\frac{\mu}{2}\|v-w\|_2^2. $
\end{assumption}

\begin{assumption}\label{Ass3}
Let $x^t_p$ be sampled from the $p$-th compromised client’s training data uniformly at random. $\mathbb{E}\| \nabla \mathcal{L}_p(\zeta^t_p, x^t_p)-\nabla \mathcal{L}_p(\zeta^t_p) \|^2 \leq \sigma _p^2$ for $p=1,2, ..., P$.
\end{assumption}

\begin{assumption}\label{Ass4}
The expectation of the squared norm of the stochastic gradient is uniformly bounded, i.e., $\mathbb{E}\| \nabla \mathcal{L}_p(\zeta^t_p, x^t_p) \|^2 \leq G^2$ for all $p=1,2, ..., P$ and $t=1,..., T-1$.
\end{assumption}

\begin{theorem}\label{The1}
Let Assumption \ref{Ass1} to \ref{Ass4} hold and $L, \mu, \sigma_p, G$ be defined therein. Specify $\xi=\frac{L}{\mu}$, $\lambda=\max\{8\xi, E\}$ and the learning rate $\alpha^t=\frac{2}{\mu(\lambda+t)}$. Training the backdoor triggers in all compromised clients satisfies:
\begin{equation}\label{Eq14}
\mathbb{E}\left[\mathcal{L}(\delta^T)\right]-\mathcal{L}^*\leq \frac{\xi}{\lambda+T-1}\left(\frac{2B}{\mu}+\frac{\mu \lambda}{2} \mathbb{E}\| \delta^1-\delta^* \|^2 \right) ,
\end{equation}
where
\begin{equation}\label{Eq15}
B=\sum_{p=1}^{P} \left(\frac{m_p}{m} \right)^2 \sigma_p^2 + 6L\Gamma +8(E-1)^2G^2.
\end{equation}
\end{theorem}

For proof, we exploit three technical lemmas from \cite{R51}.
\begin{lemma}[Results of one step SGD]\label{Lem1}
According to Assumption \ref{Ass1} and \ref{Ass2}. If $\alpha^t\leq \frac{1}{4L}$, we can get
\begin{equation}\label{Eq16}
 \mathbb{E}\| \overline{\zeta^{t+1}}-\delta^* \|^2 \leq (1-\alpha^t \mu) \mathbb{E}\| \overline{\delta^t}-\delta^*\| + (\alpha^t)^2 \mathbb{E}\|g^t-\overline{g^t}\|^2 + 6L(\alpha^t)^2\Gamma + 2\mathbb{E}\sum_{p=1}^{P}\frac{m_p}{m}\|\overline{\delta^t}-\delta_p^t \|^2,
\end{equation}
where $\Gamma =\mathcal{L}^*-\sum_{p=1}^{P}\frac{m_p}{m}\mathcal{L}_p^*\geq0$
\end{lemma}

\begin{lemma}[Bounding the variance ]\label{Lem2}
Assume Assumption \ref{Ass3} holds. Then
\begin{equation*}\label{Eq17}
\mathbb{E}\|g^t-\overline{g^t} \|^2 \leq \sum_{p=1}^{P} \left(\frac{m_p}{m} \right)^2 \sigma_p^2
\end{equation*}
\end{lemma}

\begin{lemma}[Bounding the divergence of $\{\delta^t_p\}$ ]\label{Lem3}
Assume Assumption \ref{Ass4}, that $\alpha^t$ is non-increasing and $\alpha^t \leq 2\alpha^{t+E}$ for any $t\geq 0$. There is always
\begin{equation*}\label{Eq18}
  \mathbb{E} \left[\sum_{p=1}^{P}\frac{m_p}{m}\| \overline{\delta^t} - \delta_p^t \|^2 \right]\leq 4(\alpha^t)^2 (E-1)^2 G^2
\end{equation*}
\end{lemma}

\begin{proof}
From the above definition, it always ensures $\overline{\delta^{t+1}}=\overline{\zeta^{t+1}}$. Let $\Delta^t=\mathbb{E}\|\overline{\delta^t}-\delta^* \|^2$. Based on Lemma \ref{Lem1}, Lemma \ref{Lem2} and Lemma \ref{Lem3}, we can deduce that
\begin{equation}\label{Eq19}
\Delta^{t+1} \leq (1-\alpha^t \mu)\Delta^t + (\alpha^t)^2B
\end{equation}
where
\begin{equation*}
B=\sum_{p=1}^{P}\left(\frac{m_p}{m}\right)^2 \sigma_p^2 + 6L\Gamma + 8(E-1)^2G^2.
\end{equation*}

For a non-increasing learning rate, $\alpha^t=\frac{\beta}{t+\lambda}$ for some $\beta>\frac{1}{\mu}$ and $\lambda>0$ such that $\alpha^1 \leq \min\{\frac{1}{\mu}, \frac{1}{4L} \}=\frac{1}{4L}$ and $\alpha^t \leq 2\alpha^{t+E}$. We will prove $\Delta^t \leq \frac{v}{\lambda+t}$ where $v=\max\{\frac{\beta^2B}{\beta\mu-1}, (\lambda+1)\Delta^1 \}$.

Here induction is used to prove it. Firstly, the definition of $v$ ensures that it holds for $t=1$. Assume that the conclusion holds for some $t$, then
\begin{align*}
  \Delta^{t+1} & \leq (1-\alpha^t \mu)\Delta^t + (\alpha^t)^2B \\
               & \leq \left( 1-\frac{\beta\mu}{t+\lambda} \right)\frac{v}{t+\lambda} + \frac{\beta^2B}{(t+\lambda)^2} \\
               & = \frac{t+\lambda-1}{\left(t+\lambda \right)^2}v +\left[ \frac{\beta^2B}{(t+\lambda)^2} - \frac{\beta\mu-1}{\left(t+\lambda \right)^2}v  \right] \\
               & \leq \frac{v}{t+\lambda+1}.
\end{align*}

According to the $L$-smoothness of $\mathcal{L}(\cdot)$,
\begin{equation*}\label{Eq21}
  \mathbb{E}[\mathcal{L}(\overline{\delta^t})]-\mathcal{L}^* \leq \frac{L}{2}\Delta^t \leq \frac{L}{2} \frac{v}{\lambda+t}.
\end{equation*}

Specifically, if we choose $\beta=\frac{2}{\mu}, \lambda=\max\{8\frac{L}{\mu}, E\}-1$ and denote $\xi=\frac{L}{\mu}$, then $\alpha^t=\frac{2}{\mu}\frac{1}{\lambda+t}$. It can be verified that the choice of $\alpha^t$ satisfies $\alpha^t \leq 2\alpha^{t+E}$ for $t \geq 1$. So

\begin{align*}\label{Eq22}
  v &=\max \left\{\frac{\beta^2B}{\beta\mu-1}, (\lambda + 1)\Delta^1 \right\} \\
    &\leq \frac{\beta^2B}{\beta\mu-1} + (\lambda+1)\Delta^1 \\
    &\leq \frac{4B}{\mu^2}+(\lambda+1)\Delta^1
\end{align*}
and
\begin{equation*}\label{Eq23}
  \mathbb{E}[\mathcal{L}(\overline{\delta^t})]-\mathcal{L}^* \leq \frac{L}{2}\frac{v}{\lambda+t} \leq \frac{\xi}{\lambda+t}\left(\frac{2B}{\mu}+\frac{\mu(\lambda+1)}{2}\Delta^1\right)
\end{equation*}
\end{proof}

The relevant proof of Lemmas can be found in \cite{R51}.

\section{Experiments}\label{Sec4}

This section investigates the efficacy of the proposed attack method, focusing on the following aspects: 1) Examination of factors influencing the success rate of backdoor attacks (Section \ref{Sec4.2}); 2) Comparative analysis between our proposed method and patch-based methods under the state-of-the-art federated defense (Section \ref{Sec4.3}).

\subsection{Experiment Settings}\label{Sec4.1}

\textbf{Datasets}. We evaluate the effectiveness of our method using three image datasets: CIFAR10 \cite{R58}, MNIST \cite{R59} and Mini-ImageNet \cite{R60}. The federated model and training configurations follow \cite{R1, R7}, and the detailed settings are shown in Table \ref{Tab1}. In the independent and identically distributed (IID) experiments, data is divided using random sampling, while in the non-IID experiments, the data division employs the Dirichlet distribution with hyperparameter $\alpha=0.5$. Unless specified otherwise (such as Section \ref{Sec4.2.4}), the datasets of each client are IID.

\begin{table*}[!t]
\begin{center}
\caption{Training setups for the federated model in three image datasets.}
\label{Tab1}
\footnotesize
\setlength{\tabcolsep}{1.2mm}{
\footnotesize
\begin{tabular}{c c c c c c c c c c c c}
\toprule
Datasets    & Train & Test & Image Size    & Clients & Fraction & Rounds & Local Epochs & Batch Size & Model    & Optimizer & LR \\
\midrule
CIFAR10     & 50k   & 10k  & $32\times32\times3$ & 100     & 0.1      & 500    & 5           & 50               & ResNet-18 \cite{R13} & SGD       & 0.1           \\
MNIST       & 60k   & 10k  & $28\times28\times1$ & 100     & 0.1      & 100    & 5           & 50               & ResNet-18 \cite{R13} & SGD       & 0.1           \\
Mini-ImageNet       & 48k   & 12k  & $224\times224\times3$ & 100    & 0.1      & 500   &5          &128             & ResNet-18 \cite{R13} &Adam       & 0.01           \\
\bottomrule
\end{tabular}
}
\end{center}
\end{table*}

\textbf{Backdoor Trigger Settings}. The training configurations for the backdoor trigger are detailed in Table \ref{Tab2}. The settings for poisoning in the two datasets differ slightly, primarily influenced by the convergence speed during the training process. These parameters represent optimal values derived from ablation experiments. For a comprehensive analysis, please refer to Section \ref{Sec4.2} for detailed discussions.

\begin{table*}[!t]
\begin{center}
\caption{Training setups for backdoor triggers in three image datasets.}
\label{Tab2}
\setlength{\tabcolsep}{1.2mm}{
\footnotesize
\begin{tabular}{c c c c c c c c c }
\toprule
Dataset & Trigger Size & Trigger Bound & Poison Round & Poison Ratio & Trigger Rounds & Trigger Epochs & Optimizer & LR \\
\midrule
CIFAR10  & $32\times32\times3$      & 16/255        & 200 th          & 0.3          & 40            & 10             & Adam \cite{R14}      & 0.01          \\
MNIST    & $28\times28\times1$      & 64/255        & 10 th           & 0.3          & 20            & 10             & Adam \cite{R14}      & 0.01            \\
Mini-ImageNet    & $224\times224\times3$      & 16/255        & 200 th           & 0.3         & 40            & 10            & Adam \cite{R14}      & 0.01           \\
\bottomrule
\end{tabular}
}
\end{center}
\end{table*}

\textbf{Metrics}. Our method employs the backdoor Attack Success Rate (ASR) and Main task Accuracy (MA) as the evaluation criterion. The ASR quantifies the likelihood of non-target class samples being classified into the target class after the embedding of the target trigger. MA indicators the accuracy of the global model on main (benign) machine learning tasks.

\subsection{Analysis of Trigger Factors}\label{Sec4.2}

This section explores the influence of various training parameters in our proposed MT-FBA. It considers factors such as the number of compromised clients (the poison ratio), poison interval, backdoor trigger training duration, data distribution, and the trigger bound. Each factor is individually studied through ablation experiments while maintaining the remaining parameters consistent with those outlined in Table \ref{Tab2}.

Note that the experimental methodology on the MNIST dataset is the same as that of the CIFAR10 dataset, and similar conclusions were drawn from observing the experimental outcomes. For brevity, detailed experimental results on MNIST and Mini-ImageNet are provided in \ref{Sec appendix_A} and \ref{Sec appendix_B}.

\begin{figure*}[!t]
\centering
\subfloat[Poison Ratio]{\includegraphics[width=0.19\linewidth]{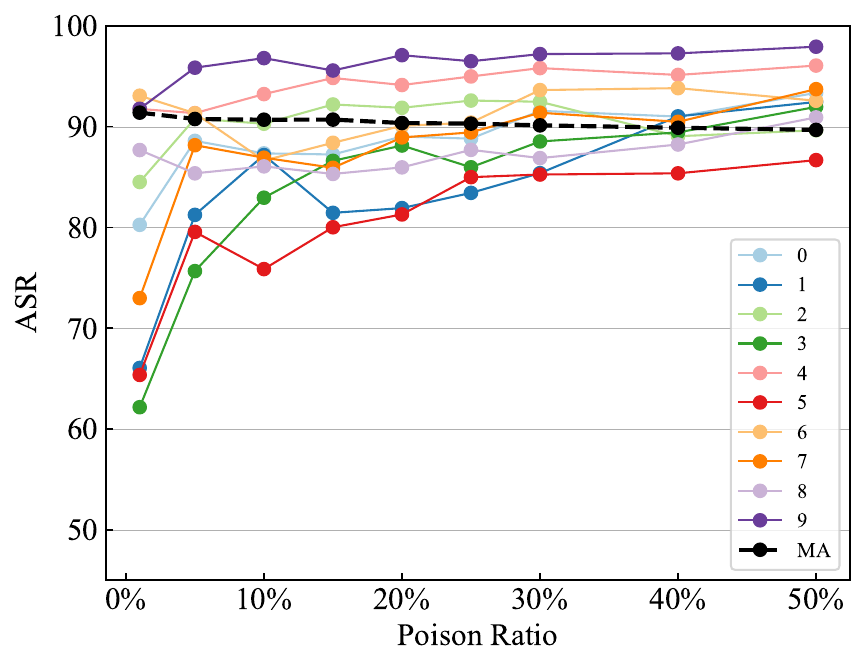}
\label{Fig3_1}}
\subfloat[Poison Point]{\includegraphics[width=0.19\linewidth]{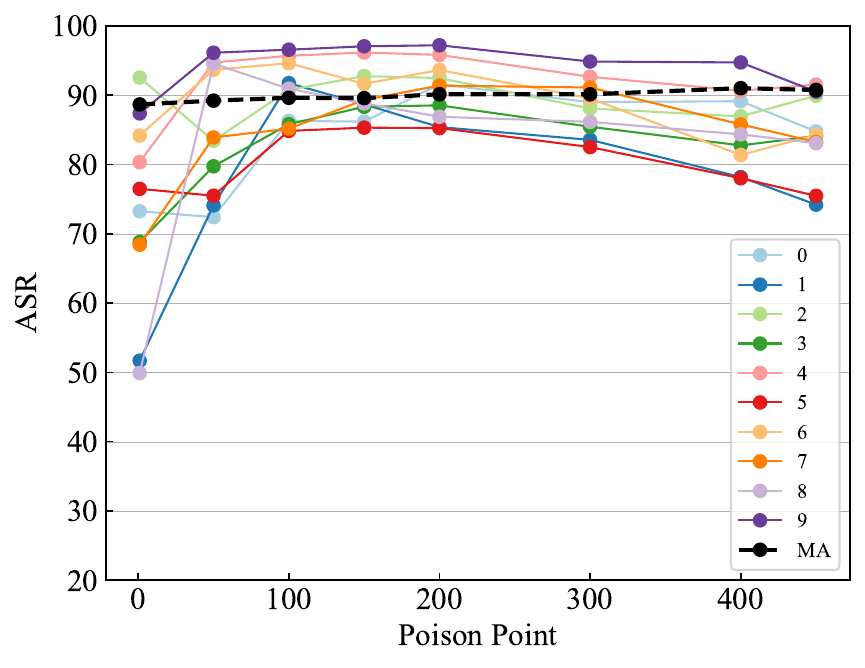}
\label{Fig3_2}}
\subfloat[Trigger Epochs]{\includegraphics[width=0.19\linewidth]{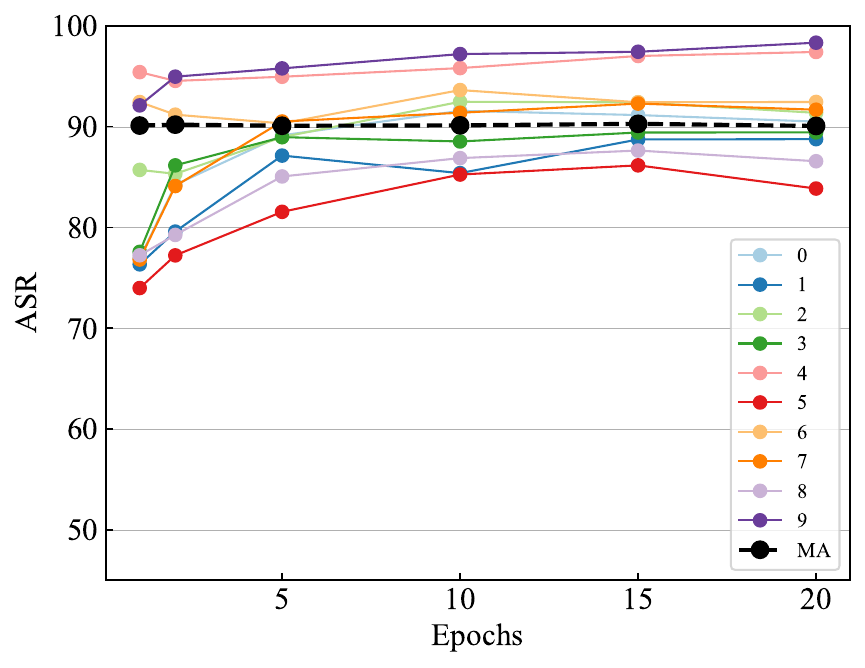}
\label{Fig3_3}}
\subfloat[Trigger Rounds]{\includegraphics[width=0.19\linewidth]{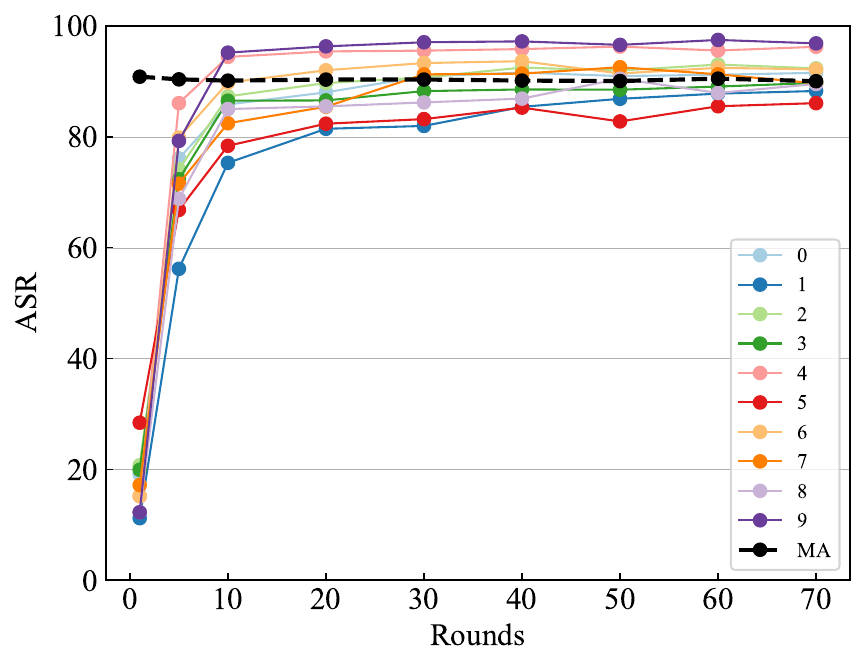}
\label{Fig3_4}}
\subfloat[Trigger Bound]{\includegraphics[width=0.19\linewidth]{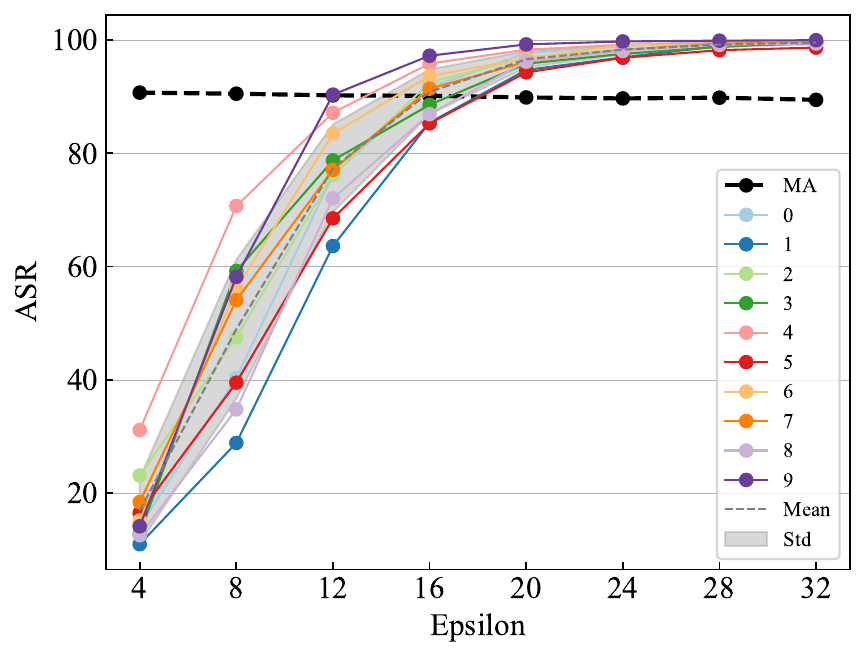}
\label{Fig3_5}}
\caption{ASR (\%) and MA (\%) for different impact factors during trigger training on the CIFAR10 dataset. Numbers represent each class of backdoor trigger.}
\label{Fig3}
\end{figure*}

\subsubsection{Effects of Poison Ratio}\label{Sec4.2.1}

The poisoning ratio, $P/K$ as mentioned in Section \ref{Sec3}, represents the ratio of compromised clients. In our series of experiments, we explored nine distinct values for $P/K$, ranging from 1\% to 50\%. These experiments were conducted using the CIFAR10 dataset, and their corresponding ASRs are depicted in Figure \ref{Fig3_1}.

The ASR for different target triggers is impacted to varying extents by the poisoning ratio. Particularly, the trigger created with target ``9" exhibits the highest attack success rate, nearly reaching 95\%. Most converge when approximately 30\% of the clients are compromised. Additionally, as the poisoning rate increases, the MA of the federated model on clean test samples slightly decreases. Therefore, we selected 30\% as the optimal poisoning rate.

Even at a low poisoning rate of 1\%, our method achieves a remarkably high ASR. Since the server randomly selects clients to participate in federated model training, most of the time, it is challenging for this compromised client to be selected. This implies that once the backdoor trigger is created, even if it is not implanted into the training set, it can be injected into the test sample during the inference phase to trigger the model's backdoor. Further discussions are provided in Section \ref{Sec5.1}.

\subsubsection{Effects of Poison Interval}\label{Sec4.2.2}

In Section \ref{Sec3.3}, backdoor trigger training starts after the near convergence of the federated model, referred to as the poisoning point. It's crucial to note that once the training data is poisoned, this sustains in all subsequent rounds. As illustrated in Figure \ref{Fig4_1}, training the federated model (ResNet-18) on the CIFAR10 dataset with IID distribution results in convergence around 200 rounds, which theoretically appears to be the optimal poisoning point.

\begin{figure}[!t]
\centering
\subfloat[IID]{\includegraphics[width=0.3\linewidth]{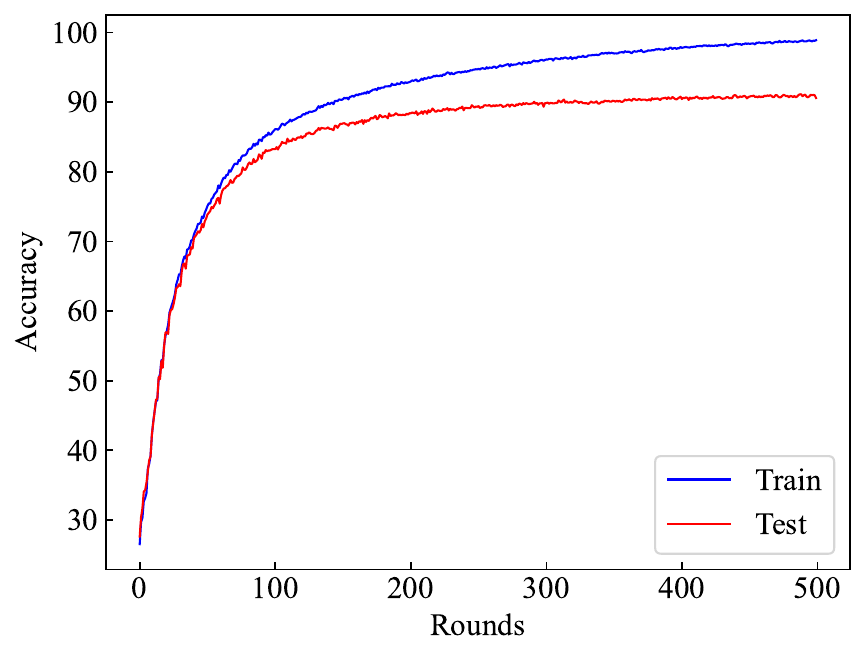}
\label{Fig4_1}}
\subfloat[non-IID]{\includegraphics[width=0.33\linewidth]{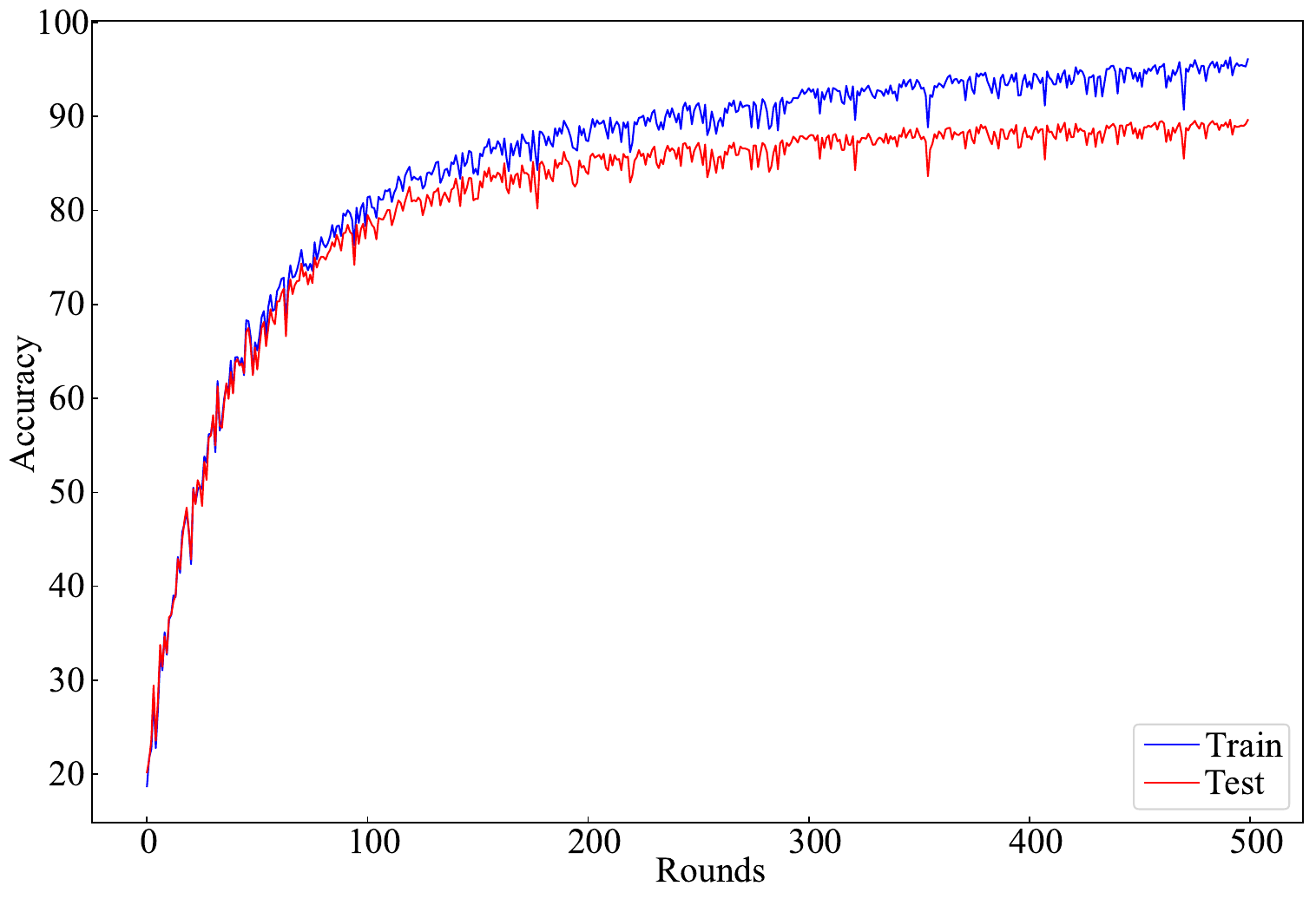}
\label{Fig4_2}}
\caption{Training the ResNet-18 model on the CIFAR10 dataset when the client data is (a) IID and (b) non-IID.}
\label{Fig4}
\end{figure}

Figure \ref{Fig3_2} presents a comparative experiment involving various poisoning points, ranging from the 1st to the 450th round. Surprisingly, poisoning at the 100th round yields a notably high backdoor attack success rate. This suggests that the federated model utilized to train the backdoor trigger only needs to be approximately close to convergence rather than trained to full convergence. Suboptimal poisoning performance in later rounds could be attributed to the model being trained for too brief a period on the poisoned data, potentially hindering its ability to retain the backdoor. Moreover, the model's accuracy on clean samples gradually increases as the poisoning point is pushed back. Therefore, the optimal poisoning point for the CIFAR10 dataset lies within the 100-200th round interval. Considering the accuracy of the global model on clean samples, we select the 200th round as the optimal poisoning point.

\subsubsection{Effects of Training Time of Trigger}\label{Sec4.2.3}

The duration of trigger training significantly impacts computational resources and communication between the compromised client and the adversary, potentially raising red flags in federated detection mechanisms. Figure \ref{Fig3_3} and  Figure \ref{Fig3_4} respectively depict the influence of local training epochs [1, 2, 5, 10, 15, 20] and communication rounds [1, 5, 10, 20, 30, 40, 50, 60, 70] on ASR. Observations indicate that epochs have a marginal impact, while the influence of communication rounds rapidly changes before stabilizing around 40 rounds. Intriguingly, this aligns with federated learning training, emphasizing the necessity for a few local epochs and more rounds to aggregate the data features of each client effectively.

\subsubsection{Effects of Trigger Bound}\label{Sec4.2.5}

In Section \ref{Sec3.3.2}, we highlight that the backdoor trigger operates within the constraints of an $\epsilon$ sphere. A smaller $\epsilon$ value limits the backdoor trigger's capacity to carry sufficient backdoor information, consequently reducing the backdoor ASR. Conversely, a larger $\epsilon$ value results in excessive modification of the original sample features by the backdoor trigger, making the backdoor sample more susceptible to detection by federated algorithms. Figure \ref{Fig3_5} illustrates the impact of varying $\epsilon$ values on the backdoor ASR, and the shaded area represents the standard deviation of the ASR for each backdoor trigger. When the trigger bound is small, the ASR for each type of backdoor trigger is low, resulting in a small standard deviation among backdoor triggers. As the trigger bound becomes larger, the ASR approaches 100\%, causing the standard deviation of the backdoor triggers to trend toward 0.  It is worth noting that as $\epsilon$ increases, the accuracy of the model on clean samples decreases. Considering the trade-off between ASR and the federated model's accuracy on clean samples, our method finally selected $\epsilon=16$, and the average backdoor ASR reached 90\%.

\subsubsection{Effects of Data Distribution}\label{Sec4.2.4}

To mitigate the potential influence of data distribution variations on experimental outcomes, the aforementioned experiments assumed that each client's data adheres to the independent and identically distributed (IID) assumption. This section explores the impact of non-IID data on the ASR. Specifically, we utilize the Dirichlet distribution as a non-IID data partitioning method with $\alpha=0.5$. Table \ref{Tab3} presents a comparison of ASR under two hypothetical scenarios on the CIFAR10 dataset, with the other parameter settings consistent with those in Tables \ref{Tab1} and \ref{Tab2}. It is observed that almost all target triggers achieve improved ASR under non-IID conditions, with an average improvement of 4.86\%. To delve deeper into the analysis, Figure \ref{Fig4_2} illustrates the federated training process under the non-IID assumption. In comparison to the training process of the IID experiment in Figure \ref{Fig4_1}, non-IID conditions introduce more oscillations and slightly lower accuracy. This should also cause the model trained under non-IID to be more vulnerable to backdoor attacks. Given that data heterogeneity in federated models is more prevalent, it serves as a reminder to consider the model's robustness to backdoor attacks when investigating federated learning with data heterogeneity.

\begin{table*}[htbp]
\begin{center}
\caption{Comparison of ASR (\%) and MA (\%) under IID and non-IID on the CIFAR10 dataset.}
\label{Tab3}
\setlength{\tabcolsep}{1.3mm}{
\footnotesize
\begin{tabular}{c c c c c c c c c c c c c}
\toprule
Data Distribution  & MA   & 0     & 1     & 2     & 3     & 4     & 5     & 6     & 7     & 8     & 9     & Average ASR \\
\midrule
IID     & 89.87 & 92.41 & 83.80 & 91.99 & 86.78 & 94.01 & 79.12 & 90.71 & 86.60 & 87.63 & 95.26 & 88.83 \\
non-IID & 86.70  &90.52 & 96.84 & 95.28 & 94.32 & 95.74 & 90.40 & 90.97 & 93.92 & 92.41 & 96.46 & 93.69   \\
\bottomrule
\end{tabular}
}
\end{center}
\end{table*}

\subsection{Compared with Patch-based Backdoor Attacks under Federated Defense}\label{Sec4.3}

As outlined in Section \ref{Sec2.3}, given the privacy of the client's local training process, current federated backdoor defense methods are typically categorized into three stages, involving operations such as filtering, truncation, and differential privacy. FLAME \cite{R25} currently stands out as the most effective federated backdoor defense method. It primarily employs clustering of local model parameters and eliminates outliers based on cosine distance. Additionally, it incorporates truncation and differential privacy to rectify backdoor models closely resembling regular local models.

The comparison results on the CIFAR10 dataset are presented in Table \ref{Tab4}, and the training settings are consistent with Table \ref{Tab1} and Table \ref{Tab2}. The ASR represents the average backdoor attack success rate obtained by all backdoor triggers. Without federated defense, DBA achieves the best ASR but significantly degrades MA. Our method maintains superior performance on both BA and MA. In the presence of FLAME, our method attains the best ASR and significantly outperforms the patch-based attack DBA. This observation also highlights the enhanced concealment capability of our method in the presence of federal defense mechanisms. It is attributed to the intra-class attack strategy and the truncated backdoor trigger, which aligns the semantic features of backdoor samples with regular samples. Consequently, the backdoor model learns a distribution similar to other regular models, avoiding detection by the FLAME defense method. In contrast, patch-based attacks, other model replacement attacks, and edge-case attacks lead the backdoor model to learn a distribution that deviates from normal samples. As a result, poison local models for other backdoor attacks are easily filtered out by the federated detection method FLAME.

\begin{table}[!t]
\begin{center}
\caption{Comparison of ASR (\%) and MA (\%) between MT-FBA and Patch-based federated backdoor attacks under the state-of-the-art federated backdoor defense.}
\label{Tab4}

\footnotesize
\begin{tabular}{l l c c c c c c c c c c c}
\toprule
\multirow{2}{*}{Attack} & \multirow{2}{*}{Dataset} & \multicolumn{2}{c}{No Defense} & \multicolumn{2}{c}{FLAME \cite{R25}}      \\
\cline{3-6}
                        &           & ASR          & MA      & ASR          & MA      \\
\midrule
Constrain-and-scale \cite{R33}      & CIFAR10    & 81.9    & 89.8         & 0        & 91.9         \\
DBA \cite{R7}                       & CIFAR10    & \underline{93.8}    & \underline{57.4}         & \underline{3.2 }     & \underline{76.2  }        \\
Edge-Case \cite{R6}                 & CIFAR10    & 42.8    & 84.3         & 4        & 79.3          \\
PGD \cite{R6}                       & CIFAR10    & 56.1    & 68.8         & 0.5      & 65.1          \\
Untargeted Poisoning \cite{R5}      & CIFAR10    & -       & 46.72        & -        & 91.31         \\
\textbf{MT-FBA (Ours)}              & CIFAR10    & \textbf{88.83}  & \textbf{89.87}   & \textbf{90.91} & \textbf{85.87} \\
\bottomrule
\end{tabular}
\end{center}
\end{table}

\begin{table}[!t]
\begin{center}
\caption{ASR (\%) of MT-FBA in the presence of input defenses on the CIFAR10 dataset.}
\label{Tab_defense}
\footnotesize
\begin{tabular}{l c c c c c }
\toprule
                   & w/o  & Bit-Red \cite{R56} & FD \cite{R55}    & NRP \cite{R54}   & DP \cite{R57}\\
\midrule
Clean Accuracy     & 90.58      & 70.42       & 29.16     & 86.04     & 37.21 \\
MT-FBA (Ours)      & 88.83      & 59.22       & 38.78    & 72.85     & 48.26 \\
\bottomrule
\end{tabular}
\end{center}
\end{table}

Additionally, we evaluated the ASR of our proposed backdoor attack method in the context of potential defenses during the federated model inference phase on the CIFAR10 dataset. As shown in Table \ref{Tab_defense}, we incorporated four defense methods: Bit-Red \cite{R56}, FD \cite{R55}, NRP \cite{R54}, and differential privacy (DP) \cite{R57}. The first three methods are commonly used in adversarial example detection tasks to filter out noise from samples, while DP is employed during the training process to protect the privacy of federated models from attacks. Table \ref{Tab_defense} demonstrates that the inclusion of any defense method reduces the model's accuracy on clean samples. Without any defense, our method achieved an ASR of 88.83\%. Although the application of defenses reduced the ASR, it still posed a significant threat to the federated model. Since our backdoor attack method introduces class-specific noise to clean samples to mislead the federated model, it is worthwhile to consider combining adversarial example detection techniques with backdoor sample detection techniques in future research to develop more robust defense methods.

\begin{figure}[!t]
\centering
\includegraphics[width=0.35\linewidth]{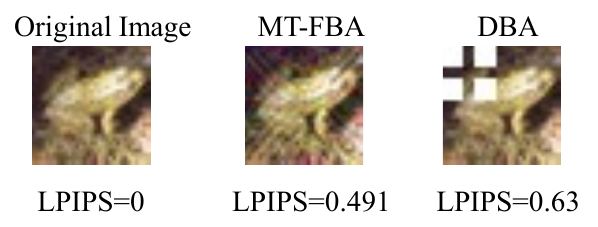}
\caption{Demonstration of a backdoor sample generated by our methods MT-FBA and contrasting method DBA on the CIFAR10 dataset. The LPIPS value represents the average of all backdoor images in the entire dataset.}
\label{FigLPIPS}
\end{figure}

It's important to note that the comparison among different backdoor attacks might be unfair due to varied methods of generating backdoor samples. Therefore, we introduce an indicator to judge concealment. Specifically, LPIPS \cite{R45} is a metric consistent with human perception that can be used to evaluate the similarity between backdoor samples and their corresponding original samples. The lower the LPIPS, the more consistent the generated image is with the original image. Backdoor samples are created on the CIFAR10 dataset employing both our method and the DBA method \cite{R7}. The comparative results are illustrated in Figure \ref{FigLPIPS}, where the numbers are the average LPIPS of all images in the entire CIFAR10 dataset. Our method yields backdoor samples with lower LPIPS values, indicating a higher degree of similarity to the original samples. This implies a potential increased difficulty in their detection. The subsequent Section \ref{Sec5.2} will elucidate the effectiveness of our proposed method through attention map visualization.

\section{Discussion}\label{Sec5}

\subsection{Zero-Shot Attak}\label{Sec5.1}
The related work \cite{R7, R9} discusses between multi-shot attacks and single-shot attacks. In a multi-shot attack, a compromised client can be repeatedly chosen by the server to participate in training across multiple rounds. Conversely, a single-shot attack involves the compromised client being selected only once, leaving the backdoor potentially weakened or forgotten by subsequent benign updates to the global model. We consider a zero-shot attack where the compromised client is either never selected or the backdoor trigger is not injected into the training set, preventing the global model from learning the backdoor trigger's features.

The experiment detailed in Section \ref{Sec4.2.1} revealed that even with a minimal proportion of poisoned clients (1\%), a relatively high backdoor attack success rate can still be achieved. This raises the question: Can our proposed method activate the model's backdoor during the inference phase even if the backdoor trigger is not injected into the training data of compromised clients? To investigate this, we conducted two independent experiments on the CIFAR10 dataset.
\begin{itemize}
  \item Train a ResNet-18 model without backdoor attacks (regular model).
  \item  Train a backdoor ResNet-18 model based on the settings outlined in Table \ref{Tab2} (multi-shot) and save the triggers.
\end{itemize}

Subsequently, inject the obtained triggers into the test sample to evaluate the ASR on both the regular model (zero-shot) and multi-shot model. The experimental results are shown in Table \ref{Tab5}.

\begin{table*}[htbp]
\begin{center}
\caption{Comparison of ASR (\%) and MA (\%) Between the Multi-shot Attack and the Zero-shot Attack on the CIFAR10 Dataset}
\label{Tab5}
\footnotesize
\begin{tabular}{c c c c c c c c c c c c c}
\toprule
           & MA   & 0     & 1     & 2     & 3     & 4     & 5     & 6     & 7     & 8     & 9     & Average ASR \\
\midrule
Multi-shot & 89.87 & 92.41 & 83.80 & 91.99 & 86.78 & 94.01 & 79.12 & 90.71 & 86.60 & 87.63 & 95.26 &  88.83   \\
Zero-shot  & 90.58 & 82.38 & 63.72 & 76.78 & 71.86 & 87.61 & 64.51 & 86.66 & 74.36 & 77.39 & 88.66 &  77.39   \\
\bottomrule
\end{tabular}
\end{center}
\end{table*}

It's remarkable to note that the average ASR of the zero-shot attack was only 11.44\% lower than that of the multi-shot attack, with the trigger created using the label ``9" achieving an ASR as high as 88.66\%. To the best of our knowledge, we are the first to propose this zero-shot federated backdoor attack. In this method, the adversary faces nearly zero risk of acquiring the model's backdoor. It only requires access to a portion of the compromised client's data and obtains the global model parameters nearing convergence at the 200th round. The success of the zero-shot attack lies in the fact that the backdoor trigger generated by our method is trained with a federated global model. Therefore, the backdoor trigger indirectly learns the overall data distribution of all clients. The introduction of Multi-Target Federated Backdoor Attack (MT-FBA) brings new insights into the defense against federated backdoors. It demonstrates that merely ensuring encryption during the communication phase is insufficient, as zero-shot attacks can directly perform backdoor attacks during the inference phase by exploiting the backdoor features embedded in the model. Hence, there is a pressing need to develop more effective detection methods for backdoor samples during the inference phase of federated models.

\subsection{Explanation via Feature Visualization}\label{Sec5.2}

Attention heatmaps serve as a valuable tool for interpreting neural networks, providing insight into the influence of target triggers by contrasting the focal points of the neural network on clean samples and backdoor samples. We use LayerCam \cite{R15} to investigate the focus of the ResNet-18 model on all backdoor triggers on the CIFAR10 dataset. The results are illustrated in Figure \ref{Fig8}.

\begin{figure*}[htbp]
\centering
\includegraphics[width=0.95\linewidth]{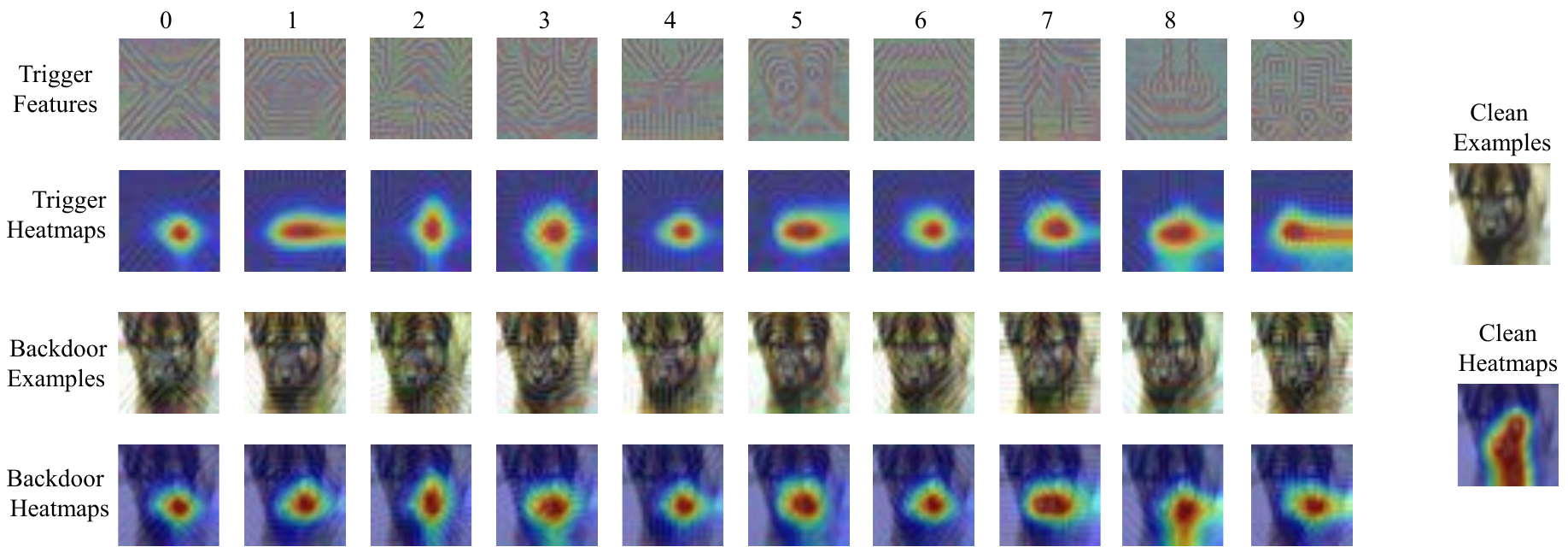}
\caption{Visualization of feature and attention heatmaps of backdoor triggers and corresponding backdoor samples on the CIFAR10 dataset. }
\label{Fig8}
\end{figure*}

Specially, the targets “0”-“9” correspond to “aircraft”, “automobile”, “bird”, “cat”, “deer”, “dog”, “frog”, “horse”, “ship”, and “truck”. Firstly, it's evident that different triggers possess distinct features. For example, trigger “8” looks like a ship, while trigger “4” mimics the corners of a deer, and so on. Additionally, each trigger exhibits inconsistent focal points. For example, triggers 1, 8, and 9 focus on lateral features and need to focus on more parts of the image, which may be related to the object's shape.

Moreover, we can observe remarkable similarities in the attention focus between the backdoor sample and the corresponding target trigger. While the original clean sample's heatmap highlights a broader area, the size, and shape of the attention area for each backdoor sample correlate with the respective trigger. For instance, the heatmaps of backdoor samples generated by triggers “0”, “4”, and “6” display a more concentrated focus, while the one from trigger “2” shows a vertical focus. This underscores how trigger features significantly alter the model's attention scope regarding the original sample feature. This variance also explains why our approach has successfully executed zero-shot attacks. This insight contributes significantly to understanding and enhancing defenses against federated backdoor attacks.

\subsection{Computational Complexity Analysis}\label{Sec5.3}

Our backdoor attack strategy involves three key steps in the training process. Step 1 is normal federated training, Step 2 is backdoor trigger training, and Step 3 mirrors the computational complexity of Step 1 as it only replaces clean samples with backdoor samples. Therefore, we focus on analyzing the computational complexity of Steps 1 and 2.

\subsubsection{Step 1: Train the Federated Model}
Assume the federated system consists of $K$ clients and one server. In each round of communication, the server selects $cK$ clients to participate in training $(0<c<1)$. Assume that there are $N$ training data samples in total, and give each client $\frac{N}{K}$ training samples according to independent and identical distribution (IID). Suppose all clients follow the same training strategy, using a batch size of $B$ and iterating for $E$ epochs. The number of gradient backpropagation steps for each client in one round of training is $\frac{EN}{KB}$. Therefore, in one round of communication, the total number of gradient backpropagation steps for all participating clients is $\frac{EN}{KB} \ast cK=\frac{cEN}{B}$. If the federated system performs a total of $R$ communication rounds, the total number of gradient backpropagation steps is $\frac{cREN}{B}$. Thus, the number of gradient backpropagations is proportional to the proportion $c$ of clients participating in federated training.

\subsubsection{Step 2: Train Backdoor Triggers}
Assume there are $P$ compromised clients among the $K$ clients. Each compromised client has the same number of training samples as clean clients, i.e., $\frac{N}{K}$. Assume each compromised client uses a batch size of $B'$, iterates for $E'$ epochs to train its backdoor trigger, and needs $R'$ rounds of communication to aggregate the features of the backdoor trigger. The total number of gradient backpropagation steps for all compromised clients to train the backdoor trigger is $\frac{PR'E'N}{KB'}$. Therefore, the number of gradient backpropagations is proportional to the ratio of compromised clients $\frac{P}{K}$.

The above analysis focuses solely on the client training process, which requires GPUs for gradient computation. The aggregation process on the server, which uses CPUs, is not included in this complexity analysis. Server aggregation involves summing and averaging the parameters of all client models sequentially, as described by Eq. (\ref{Eq1}). Similarly, the attacker uses an analogous operation to aggregate the features of the backdoor triggers from all compromised clients, as detailed in Eq. (\ref{Eq9}). The computational complexity of model aggregation is proportional to the size of model parameters $w$ and the number of clients $cK$ participating in aggregation.  On the other hand, the computational complexity of backdoor trigger aggregation is proportional to the size and categories of the backdoor trigger $\delta$, as well as the number of compromised clients $P$. However, since model or trigger aggregation is done on the server (which is, by default, a device with powerful computing power), this time increase is minimal. In addition, the number of parameters involved in model or trigger aggregation is very small compared to the amount of gradient calculations during training, so the aggregation time is negligible relative to the overall training time.  Ablation experiments on different model architectures and client numbers are conducted on the CIFAR10 dataset. The test system is configured with an NVIDIA 3090 GPU, an Intel(R) Xeon(R) Silver 4214 CPU @ 2.70 GHz, and Ubuntu 18.04.5 LTS.

\begin{table}[!t]
\begin{center}
\caption{ASR ($\%$), MA ($\%$), and Training Time ($h$) for Different Numbers of Clients on the CIFAR10 Dataset}
\label{Tab6}
\footnotesize
\begin{tabular}{c c c c}
\toprule
Client Number & MA ($\%$) & ASR ($\%$) & Time ($h$) \\
\midrule
100          & 89.87   & 88.83    & 10.15    \\
200          & 87.71   & 81.49    & 10.00    \\
300          & 86.36   & 81.68    & 10.61    \\
500          & 79.82   & 73.19    & 10.50    \\
800          & 78.34   & 78.18    & 10.62    \\
1000         & 70.19   & 71.97    & 10.64    \\
\bottomrule
\end{tabular}
\end{center}
\end{table}

\begin{table}[!t]
\begin{center}
\caption{ASR ($\%$), MA ($\%$), and Training Time ($h$) for Different Models on the CIFAR10 Dataset}
\label{Tab7}
\footnotesize
\begin{tabular}{c c c c c}
\toprule
Model   & MA ($\%$)   & ASR ($\%$)  & Parameters (M) & Time ($h$) \\
\midrule
ReaNet-18           & 89.87 & 88.83 & 11.17      & 11.5     \\
ResNet-50           & 89.19 & 90.35 & 23.52      & 23.0     \\
ResNet-101          & 89.50 & 85.99 & 42.51      & 28.0     \\
Inception-v3        & 79.14 & 60.82 & 21.81      & 25.0     \\
InceptionResNet-v2  & 88.57 & 63.93 & 54.32      & 44.5     \\
\bottomrule
\end{tabular}
\end{center}
\end{table}

First, as shown in Table \ref{Tab6}, we varied the number of clients to 100, 200, 300, 500, 800, and 1000 to verify the impact of the federated system's scale on computational complexity. Other parameters are shown in Table \ref{Tab1} and Table \ref{Tab2}.  Experiments show that as the number of clients increases, the main task accuracy (MA) of the model decreases, and the backdoor attack success rate (ASR) also decreases, but the computing time does not increase too much. It indicates that the factors affecting the system operation speed are mainly in the client training process rather than the server aggregation process. Second, as shown in Table \ref{Tab7}, we conducted an ablation analysis on different classification models, including ResNet-18, ResNet-50, ResNet-101, Inception-v3 \cite{R48}, and InceptionResNet-v2 \cite{R49}. The experimental metrics include the main task accuracy (MA), the backdoor attack success rate (ASR), the number of parameters, and training time. Even if different models have similar accuracy, they may exhibit different backdoor attack success rates due to varying levels of robustness inherent in different model architectures. In the field of adversarial examples research, there are many discussions on the robustness of different model architectures, such as \cite{R46, R47}. Additionally, as the number of model parameters increases, the training time may also increase, although this is not always directly proportional. The training speed of deep learning models is also influenced by factors beyond just the number of parameters, including the model architecture and inference characteristics such as FLOPs (Floating Point Operations), which are beyond the scope of this work.

\section{Conclusions}\label{Sec6}

The paper introduces a novel federated backdoor attack approach based on feature aggregation. Our approach addresses key limitations of patch-based federated attack methods by aligning trigger size with the image, eliminating the need for manual design of patch shape and position. Additionally, it incorporates feature interactions among triggers from multiple compromised clients during training. Leveraging the intra-class attack strategy, we demonstrate the simultaneous generation of backdoor triggers for all target classes. Experimental results show the robustness of our method against federated defense mechanisms and its capability to execute zero-shot attacks. Moreover, visualizing the distinctive characteristics of backdoor triggers offers insights into defense against backdoor attacks.

Presently, our method still requires the simultaneous utilization of the training data from the compromised client and the nearly converged global model to train backdoor triggers. This requires a high level of authorization for the compromised client. In future work, we aim to achieve a federated backdoor attack by only utilizing the training data of the compromised client to train triggers, without implanting the triggers into the training data in subsequent rounds. This represents a model-agnostic zero-shot attack.

\section{Acknowledgements}

This work was supported in part by the Fundamental Research Funds for the Central Universities (2232021A-10), National Natural Science Foundation of China (61903078), Shanghai Pujiang Program (22PJ1423400), Shanghai Sailing Program (no. 22YF1401300), Natural Science Foundation of Shanghai (21ZR1401700).

\appendix
\section{Experiments on the MNIST dataset}\label{Sec appendix_A}

\subsection{Ablation Study of Poison Ratio and Interval}

In our exploration of the poisoning ratio, we examined nine distinct values for $P/K$, ranging from 1\% to 50\%. As shown in Table \ref{Tab2}, we set the boundary of the backdoor trigger on the MNIST dataset to 64/255 to clearly show the relationship between each influencing factor and the success rate of the backdoor attack. As depicted in Figure \ref{Fig9_1}, when the poisoning rate is lower than 30\%, the attack success rate of half of the backdoor triggers is low. Notably, the backdoor trigger labeled as ``8" demonstrates the highest backdoor attack success rate. When considering the performance of the backdoor triggers associated with labels ``0", ``4", ``6", ``8", and ``9", the ASR stabilizes at around a 30\% poisoning rate, and the MA exhibits marginal decline. Figure \ref{Fig9_2} reveals that the poisoning begins around the 10th round to achieve the optimal ASR, with minimal impact on the MA.

\begin{figure*}[htpb]
\centering
\subfloat[Poison Ratio]{\includegraphics[width=0.19\linewidth]{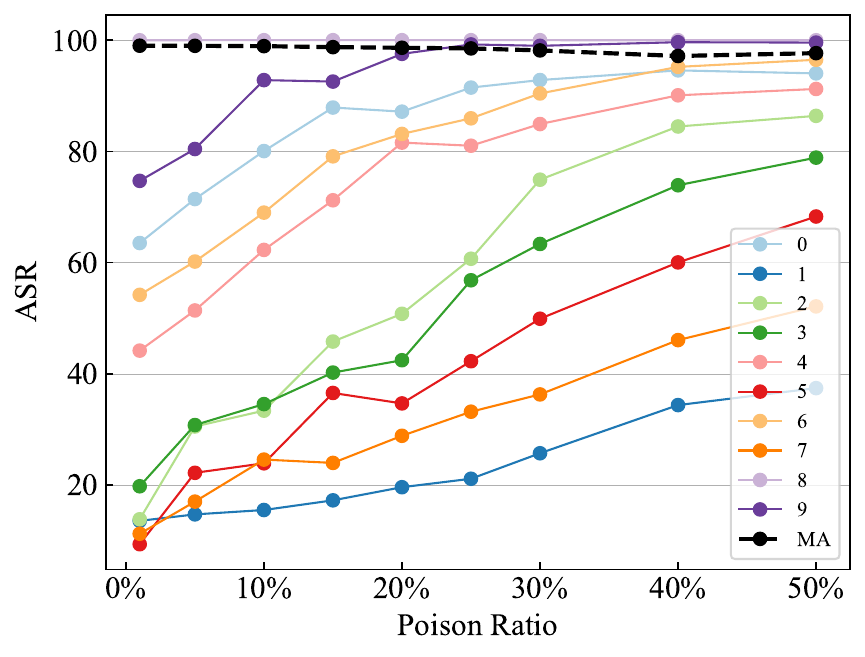}
\label{Fig9_1}}
\subfloat[Poison Point]{\includegraphics[width=0.19\linewidth]{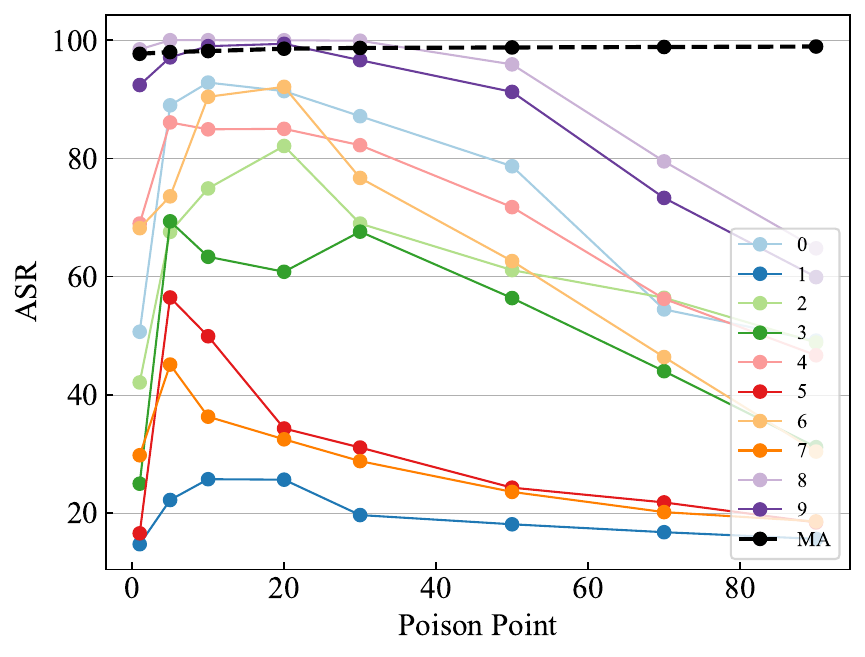}
\label{Fig9_2}}
\subfloat[Trigger Epochs]{\includegraphics[width=0.19\linewidth]{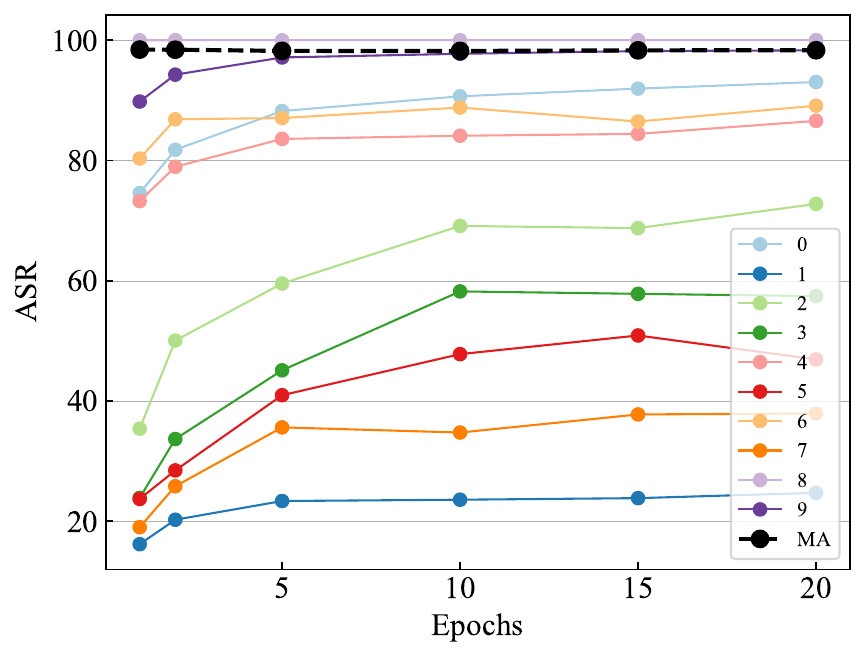}
\label{Fig9_3}}
\subfloat[Trigger Rounds]{\includegraphics[width=0.19\linewidth]{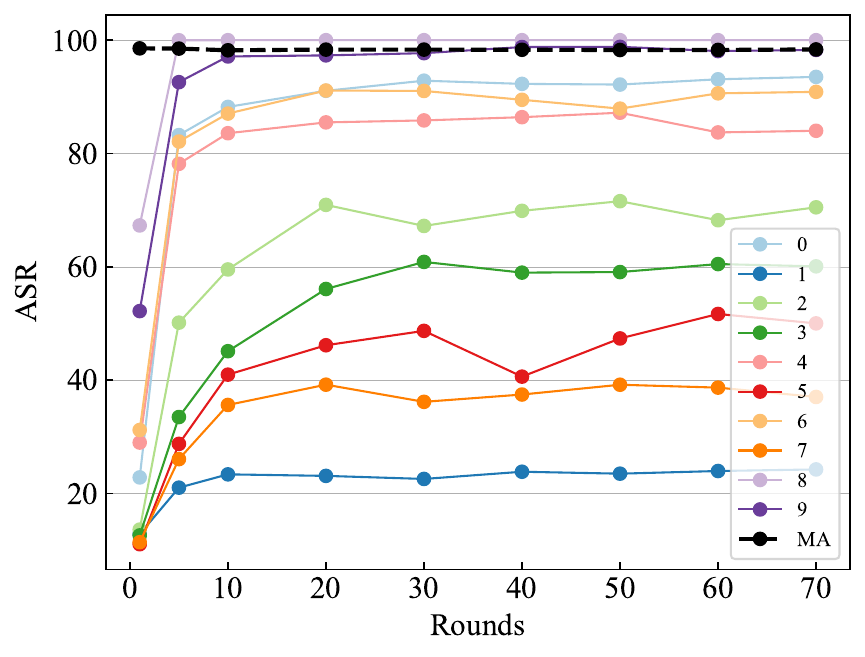}
\label{Fig9_4}}
\subfloat[Trigger Bound]{\includegraphics[width=0.19\linewidth]{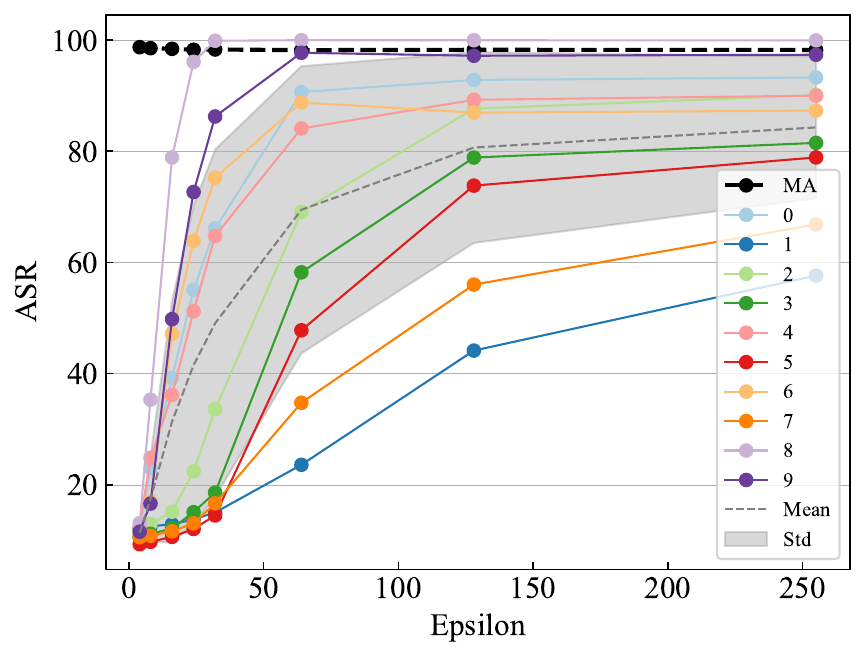}
\label{Fig9_5}}

\caption{ASR (\%) and MA (\%) for different impact factors during trigger training on the MNIST dataset. Numbers represent each class of backdoor trigger.}
\label{Fig9}
\end{figure*}

\subsection{Ablation Study of Trigger Epochs and Rounds}

Given the quick convergence of model training on the MNIST dataset, the training duration for the backdoor trigger is relatively short. As indicated in Figure \ref{Fig9_3} and Figure \ref{Fig9_4}, achieving the optimal backdoor attack success rate necessitates approximately 10 epochs and 20 rounds for training the backdoor triggers. The outcomes of the backdoor attack on the MNIST dataset align with the conclusions drawn from the experiments on the CIFAR10 dataset. Specifically, the number of epochs dedicated to training backdoor triggers on the compromised clients has a minor impact on the ASR. Conversely, the aggregation of features among compromised clients, denoted by the rounds, holds more significance. This again emphasizes the crucial role of feature aggregation of backdoor triggers among compromised clients in our proposed approach.

\subsection{Ablation Study of Trigger Bounds}

The ASR demonstrates a positive correlation with the pixel-bound value $\epsilon$ of the backdoor trigger, indicating that a larger $\epsilon$ leads to a higher ASR. As illustrated in Figure \ref{Fig9_5}, the shaded area represents the standard deviation of the ASR for each backdoor trigger. The standard deviation of ASR for each backdoor trigger on the MNIST dataset is larger, which is not obvious on the CIFAR10 dataset. The ASR achieved by the backdoor triggers associated with labels ``0", ``4", ``6", ``8", and ``9" on the MNIST dataset displays a more substantial variation with increasing $\epsilon$, whereas other backdoor triggers yield low ASR. We guess that this discrepancy arises due to the simplistic nature of features in the MNIST dataset, particularly in numbers like ``1" and ``7". An observable trend is that the backdoor triggers linked to numbers containing the circular feature, such as ``8", perform notably well. These numbers represent relatively complex features.

\begin{figure}[!t]
\centering
\includegraphics[width=0.4\linewidth]{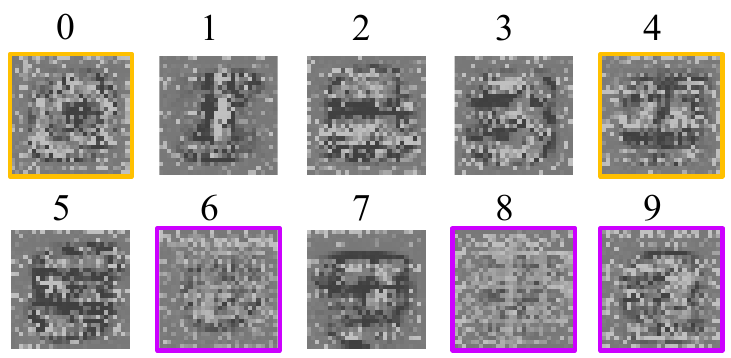}
\caption{Visualization of backdoor trigger features on the MNIST dataset.}
\label{Fig12}
\end{figure}

The triggers corresponding to each class are visually depicted in Figure \ref{Fig12}. We can observe that backdoor triggers labeled as ``6", ``8", and ``9" (outlined in purple) exhibit the most features or perturbations (depicted in white areas), while triggers ``0" and ``4" (outlined in orange) possess relatively fewer features, and ``1", ``5", and ``7" have the least features. When a backdoor trigger contains fewer distinctive features, embedding it into samples of other classes becomes challenging, as it struggles to alter the original features of samples significantly. Therefore, it is very difficult to make the classification model classify backdoor samples into the class consistent with backdoor triggering. This observation aligns with the findings from experiments conducted on the CIFAR10 dataset. As depicted in Figure \ref{Fig8}, the backdoor trigger labeled as ``9" has the most diverse features, consistently leading to the highest backdoor attack success rates in numerous CIFAR10 experiments (e.g., Figure \ref{Fig3}). However, the CIFAR10 dataset showcases relatively rich features within each class, alleviating the gap between classes. Conversely, in the MNIST dataset, the backgrounds are simpler, amplifying the distinctions between classes.

\section{Experiments on the Mini-ImageNet dataset}\label{Sec appendix_B}

\subsection{Ablation Study of Trigger Bounds}

\begin{figure}[htpb]
\centering
\includegraphics[width=0.35\linewidth]{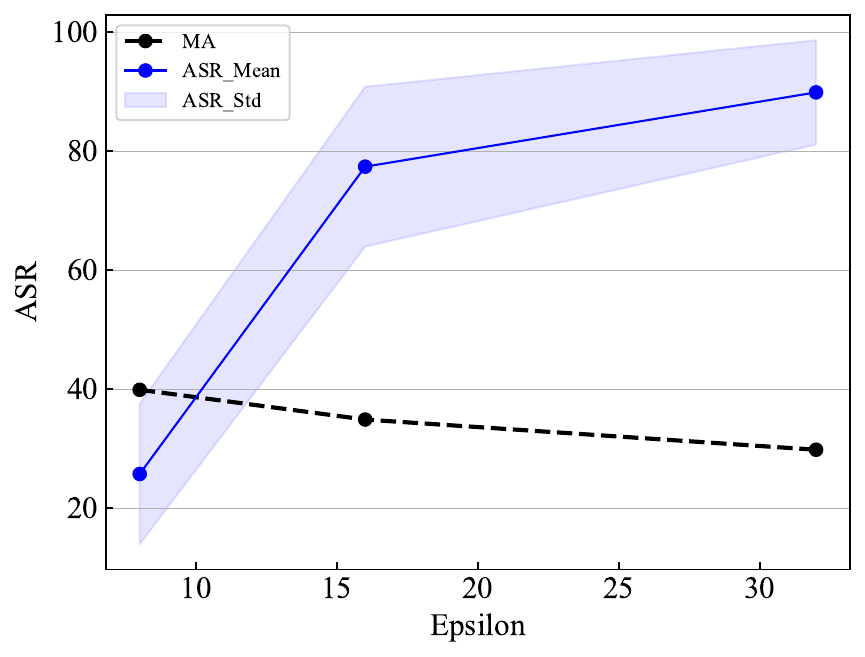}
\caption{ASR (\%) and MA (\%) for different trigger bounds during trigger training on the Mini-ImageNet dataset.}
\label{Fig13}
\end{figure}

Mini-ImageNet \cite{R60} is a subset extracted from the ImageNet dataset. It contains 100 classes with 600 images per class. For our experiments, we cropped the image size to 224*224*3. Because it was originally designed for few-shot learning tasks, we re-divided it into a training set and a test set with an 8:2 ratio. The relevant settings of the dataset are shown in Table \ref{Tab1}, and the backdoor attack design is detailed in Table \ref{Tab2}. To investigate the relationship between the backdoor trigger bound and the attack success rate (ASR), we set three bounds: 8, 16, and 32 (image pixel values range from 0 to 255). The experimental results are presented in Figure \ref{Fig13}. The black dotted line represents the accuracy on clean samples, the blue solid line shows the average backdoor ASR and the shaded area indicates the standard deviation of the ASR for each class of backdoor trigger (a total of 100 categories). We observe that as the backdoor trigger bound increases, the ASR also increases, while the variance in the ASR for each backdoor trigger remains relatively stable. Compared to the results on CIFAR10 (Figure \ref{Fig3_5}) and MNIST (Figure \ref{Fig9_5}), the experiment on the Mini-ImageNet dataset further confirms that ASR increases with the bound. Additionally, by comparing these figures, we find that for the MNIST dataset, which has simpler features, the variance in ASR across different backdoor triggers is larger. In contrast, for the CIFAR10 and Mini-ImageNet datasets, which have more complex features, the variance in ASR is smaller.



\bibliographystyle{elsarticle-num} 
\bibliography{MT-FBA}



\end{document}